\documentclass[12pt, draftclsnofoot,onecolumn]{IEEEtran}
\usepackage{color}     
\usepackage{epsf,psfrag,amssymb,amsfonts,latexsym,cite,verbatim,enumerate,subfigure}
\usepackage{srcltx,amsmath,cases, graphicx}
\usepackage{times, multirow, array}
\usepackage[mathscr]{eucal}

\def\psfancypar#1#2{\begingroup\def\par{\endgraf\endgroup\lineskiplimit=0pt}
               \setbox2=\hbox{\large\sc #2}
               \newdimen\tmpht \tmpht \ht2 \advance\tmpht by \baselineskip
               \font\hhuge=Times-Bold at \tmpht
               \setbox1=\hbox{{\hhuge #1}}
               \count7=\tmpht \count8=\ht1
               \divide\count8 by 1000 \divide\count7 by \count8 
               \tmpht=.001\tmpht\multiply\tmpht by \count7 
               \font\hhuge=Times-Bold at \tmpht
               \setbox1=\hbox{{\hhuge #1}}
               \noindent
                \hangindent1.05\wd1
               \hangafter=-2 {\hskip-\hangindent
               \lower1\ht1\hbox{\raise1.0\ht2\copy1}%
                \kern-0\wd1}\copy2\lineskiplimit=-1000pt}

\newcommand{\E}{\mbox{{\rm E}}}

\newcommand{\abf}{\mbox{${\bf a}$}}

\def\boxit#1{\vbox{\hrule\hbox{\vrule\kern3pt
        \vbox{\kern3pt#1\kern3pt}\kern3pt\vrule}\hrule}}

\def\reals{ { {\rm  I \kern-0.15em R }  } }
\def\complex{ {\,{{\rm C} \kern-0.50em \raise0.20ex {  |}}\, }}

\def\mubf{\hbox{\boldmath$\mu$\unboldmath}}

\def\Sigmabf{\hbox{$\bf \Sigma$}}

\def\abf{{\bf a}}

\def\ebf{{\bf e}}

\def\hbf{{\bf h}}

\def\ubf{{\bf u}}

\def\wbf{{\bf w}}
\def\xbf{{\bf x}}

\def\xbf{{\bf x}}

\def\Abf{{\bf A}}

\def\Ibf{{\bf I}}

\def\Rbf{{\bf R}}

\def\Cc{{\cal C}}

\def\Nc{{\cal N}}

\def\Rc{{\cal R}}
\def\Sc{{\cal S}}

\def\be{\vskip .3cm \begin{equation}}
\def\ee{\end{equation} \vskip .4cm \noindent}
\newcommand{\R}{\mbox{$\hat {\bf R}_{N}$}}

\def\Rxx{\Rbf_{\ssstyle X\kern-.1em X}}

\let\ssstyle=\scriptscriptstyle

\def\Kout{\setbox1=\hbox{\Huge\bf K}\hbox to
1.05\wd1{\hspace{.05\wd1}
}}
\def\Sout{\setbox1=\hbox{\Huge\bf S}\hbox to 1.05\wd1{\hspace{.05\wd1}
}}

  \ifx\LabelFigloaded\MYundefined\relax
  \else
   \fi

  \def\LabelFigloaded{\relax}

  \chardef\LabelFigCatAt\the\catcode`\@
  \catcode`\@=11

 \let\LabelFigwlog@ld\wlog
 \def\wlog#1{\relax}

 \ifx\\\MYundefined@
    \let\\\relax
 \fi

  \def\ms@g{\immediate\write16}

 \def\N@wif{\csname newif\endcsname }
 \def\Temp@ {\N@wif\ifIN@}
 \ifx\INN@\MYundefined@
    \else \let\Temp@\relax
 \fi
 \Temp@

  \def\IN@{\expandafter\INN@\expandafter}
  \long\def\INN@0#1@#2@{\long\def\NI@##1#1##2##3\ENDNI@
    {\ifx\m@rker##2\IN@false\else\IN@true\fi}%
     \expandafter\NI@#2@@#1\m@rker\ENDNI@}
  \def\m@rker{\m@@rker}
 
  \newtoks\Initialtoks@  \newtoks\Terminaltoks@
  \def\SPLIT@{\expandafter\SPLITT@\expandafter}
  \def\SPLITT@0#1@#2@{\def\TTILPS@##1#1##2@{%
     \Initialtoks@{##1}\Terminaltoks@{##2}}\expandafter\TTILPS@#2@}

 \def\Shifted@@#1#2#3{\setbox0=\hbox{#3}%
   \raise -\dp0\vbox {\kern-#2%
       \hbox {\kern#1\unhbox0\kern-#1}%
           \kern#2}}

 \newcount\gridcount
 \newbox\auxGridbox@ \newbox\hGridbox@ \newbox\vGridbox@
 \newbox\Labelbox@ \newbox\auxLabelbox@
 \newbox\Coordinatebox@
 \newtoks\Labeltoks@
 \newdimen\Wdd@ \newdimen\Htt@
 \newdimen\Wddd@ \newdimen\Httt@
 
 \def\Wr@{\immediate\write16}

 \newdimen\GL@wd
 \def\GridLineWidth#1{\GL@wd=#1}

 \def\gobble#1{}
 \def\EdgeErr@{\Wr@{}%
      \Wr@{\string\Edges\space argument
      1, 10, 100 or 1000 please\string!}%
      }

 \newcount\Edgect@

 \def\Sweepup#1\endSweepup{}

 \def\SetEdges@{%
    \edef\Zr@@s{\expandafter\gobble\number\Edgect@\empty}%
        \count255=0\Zr@@s\relax
        \ifnum\count255=\z@\else\EdgeErr@\show\tailtest\fi
        \count255=1\Zr@@s\relax
        \ifnum\count255=\Edgect@\relax\else\EdgeErr@\show\leadtest\fi
    \EdgGl@b\edef\Zr@s{\expandafter\gobble\Zr@@s\empty}
    \ifnum\Edgect@>\@ne\relax\EdgGl@b\let\L@Dc\empty
        \else\EdgGl@b\edef\L@Dc{\string.}\fi
    \ifnum\Edgect@>\@ne\relax
        \EdgGl@b\edef\Edgescale@##1{\divide##1 by \Edgect@}%
        \else\EdgGl@b\edef\Edgescale@##1{}\fi
    }

 \def\Edges#1{\Edgect@=#1\relax
     \let\EdgGl@b\global \SetEdges@}

 \Edges{1}

 \def\hhrule{\hrule height \GL@wd\vskip-.\GL@wd}

 \def\hRule@{%
   \advance\gridcount -2%
   \vfil\hhrule\vfil
   \llap{\smash{\raise -2.5pt
     \hbox{\L@Dc\number\gridcount\Zr@s\kern2pt}}}%
   \hhrule
   }

\def\vvrule{\vrule width \GL@wd \kern-\GL@wd}

 \def\vRule@{\advance\gridcount 2%
   \hfil\vvrule\hfil
   \setbox\auxGridbox@=\vbox to 0pt
      {\vskip \Htt@\vskip 2pt
        \hbox to 0pt{\hss\L@Dc\number\gridcount\Zr@s\hss}\vss}%
      \wd\auxGridbox@=0pt \box\auxGridbox@
   \vvrule
   }

 \def\PlaceGrid@@{\gridcount=10 
  \setbox\hGridbox@=\hbox{%
        \hbox{%
             \hskip-.4pt\vrule
             \vbox to \Htt@{%
               \offinterlineskip\parindent=\z@\relax
               \hbox to \Wdd@{\hfil}
               \hRule@\hRule@\hRule@\hRule@
               \vfil\hhrule\vfil}%
             \vrule\hskip-.4pt}
    }%
  \gridcount=0%
  \setbox\vGridbox@=\hbox{%
      \vbox{\offinterlineskip\parindent=0pt\hsize=0pt
         \vskip-.4pt\hrule%
         \hbox to \Wdd@{%
                 \vtop to \Htt@{\vfil}%
                 \vRule@\vRule@\vRule@\vRule@
                 \hfil\vvrule\hfil}%
         \hrule\vskip-.4pt}}%
  \wd\hGridbox@=0pt\ht\hGridbox@=0pt
  \wd\vGridbox@=0pt\ht\vGridbox@=0pt
  \hbox{\box\hGridbox@\box\vGridbox@}%
  }

 \def\LabelsGlobal{\def\LabGl@b{\global}}
 \def\LabelsLocal{\def\LabGl@b{}}
 \LabelsGlobal 

 \def\SetLabels#1\endSetLabels{%
   \LabGl@b\Labeltoks@={#1()\\}%
   }

 \LabGl@b\Labeltoks@={()\\}

 \def\ShowGrid{\LabGl@b\let\PlaceGrid@\PlaceGrid@@}
 \def\HideGrid{\LabGl@b\let\PlaceGrid@\relax}
 \def\Grids{\ShowGrid\LabGl@b\let\GridSwitch@\ShowGrid}
 \def\noGrids{\HideGrid\LabGl@b\let\GridSwitch@\HideGrid}

 \noGrids

 \def\bAdjust@@{%
     \setbox\auxLabelbox@=\hbox{\raise \dp\auxLabelbox@
            \box\auxLabelbox@}}
 \def\bAdjust@{\let\vAdjust@\bAdjust@@}

 \def\eAdjust@@{\dimen0=-.5\ht\auxLabelbox@
     \advance\dimen0 by .5\dp\auxLabelbox@
     \setbox\auxLabelbox@=
            \hbox{\raise\dimen0\box\auxLabelbox@}}
 \def\eAdjust@{\let\vAdjust@\eAdjust@@}

 \def\tAdjust@@{%
     \setbox\auxLabelbox@=\hbox{\raise-\ht\auxLabelbox@
            \box\auxLabelbox@}}
 \def\tAdjust@{\let\vAdjust@\tAdjust@@}

 \let\vAdjust@\relax

 \def\lAdjust@{\let\hAdjust@\rlap}
 \def\rAdjust@{\let\hAdjust@\llap}

 \let\hAdjust@\relax\let\vAdjust@\relax

 \def\FetchLabel@#1(#2)#3\\{%
     \IN@0#2@@\ifIN@
        \setbox0=\hbox{\ignorespaces#1#3\unskip}%
        \ifdim\wd0>0pt
           \ms@g{}%
           \ms@g{ !!! Bad label(s)? !!!}%
           \message{ #1(#2)#3}%
        \fi
        \def\LabelMole@##1\endFetchLabel@{%
            \IN@0()\\@##1@%
            \ifIN@\def\Temp@{\FetchLabel@##1\endFetchLabel@}%
            \else\def\Temp@{}%
            \fi
            \Temp@
           }%
     \else
       \ignorespaces#1\unskip
       \setbox\auxLabelbox@=%
         \hbox to 0pt{\hss\ignorespaces\hAdjust@
          {\ignorespaces#3\unskip}\hss}%
       \vAdjust@
       \let\hAdjust@\relax\let\vAdjust@\relax
       \AugmentLabelBox@@{#2}%
       \ht\Labelbox@=0pt\dp\Labelbox@=0pt
       \let\LabelMole@\FetchLabel@%
     \fi\LabelMole@}

 \newtoks\XYSep@ 
 \def\SetXYSeparator#1{%
     \IN@0#1@@\ifIN@\XYSep@{*}%
     \else
     \XYSep@{#1}%
     \fi
     }

 \SetXYSeparator*

 \def\AugmentLabelBox@@#1{%
     \IN@0\the\XYSep@ @#1@\ifIN@
       \SPLIT@0\the\XYSep@ @#1@%
       \setbox\Labelbox@=\hbox to 0pt{%
         \unhbox\Labelbox@
         \Shifted@@{\the\Initialtoks@\Wddd@}%
         {\the\Terminaltoks@\Httt@}%
         {\box\auxLabelbox@}}%
     \else
         \ms@g{}%
         \ms@g{ !!! Bad insertion point. !!!}%
         \message{ (#1\ this point was rejected.)}%
     \fi
    }

 \def\FetchOption@#1[#2]#3\endFetchOption@{%
    \def\temp{#1}
    \ifx\temp\empty
       \Edgect@=#2\relax
       \let\EdgGl@b\relax
       \SetEdges@
       \Cleaner@#3%
    \fi}

 \def\Cleaner@#1[@]{\Labeltoks@{#1}}
     
 \def\PlaceLabels@@{\mathsurround=0pt
     \def\Cr@{\\}%
     \let\L\lAdjust@\let\R\rAdjust@
     \let\B\bAdjust@\let\E\eAdjust@\let\T\tAdjust@
     \expandafter\FetchOption@\the\Labeltoks@[@]\endFetchOption@
     \Wddd@=\Wdd@ \Edgescale@\Wddd@ 
     \Httt@=\Htt@ \Edgescale@\Httt@
     \expandafter\FetchLabel@\the\Labeltoks@\endFetchLabel@
     \box\Labelbox@
     }%

 \let \PlaceLabels@\PlaceLabels@@

 \def\AffixLabels#1{\setbox\Coordinatebox@=\hbox{#1}%
      \Wdd@=\wd\Coordinatebox@ \Htt@=\ht\Coordinatebox@
      \advance\Htt@ \dp\Coordinatebox@
      \hbox{\copy\Coordinatebox@\kern-\Wdd@ 
           \Shifted@@{0pt}{-\dp\Coordinatebox@}%
           {\PlaceLabels@\PlaceGrid@}%
           \kern\Wdd@}%
      \GridSwitch@ 
      \LabGl@b\Labeltoks@{()\\}%
      }
 
   \let\wlog\LabelFigwlog@ld   
   \catcode`\@=\LabelFigCatAt  

                                By

              Raymond S\'eroul <A18645@FRCCSC21.BITNET>
                                and 
              Laurent Siebenmann <lcs@topo.math.u-psud.fr>
    
              VERSIONS: July 1991, Oct 1991, Jan 1992, July 1992

INTRODUCTION

      This labelling package is intended for TeX users who
rely on non-TeX sources for for their graphics inserts.  It
provides means for adding TeX labels to such inserts with a
minimum of fuss. 

       For most labels, TeX users have in the past found it
reasonably convenient to rely on non-TeX sources. Typical
occasions when an inescapable need for TeX labels seemed to
arise are

 (a) when the graphics program lacks certain exotic or complex
mathematical symbols

 (b) when the very highest typographical quality is wanted for the
labels

 (c) when labels included with the graphics fail to print, 
 and you cannot figure out why (cf. boxedeps.doc).  The labels
 provided by labelfig.tex are 100

       Since this package first appeared, many users, who in the
past scarcely dreamed of using TeX labels, have come to use
nothing but.  So it is now appropriate to add

Intoxication Warning:  TeX labels may be addictive and expensive. 

     If you have a fast preview you may disagree, and even find
that this package provides an agreeable paste-up environment; see
extra applications at end.

     Note to publishers: It is possible and convenient to ultimately
export the TeX labels produced by labelfig.tex to become an integral
part of the EPS file. This is often desired by a publisher who typically
uses an "upmarket" graphics or page layout program, with which the
staff is skilled in perfecting figures.  See Appendix I for
a recipe.

     The authors are grateful to Patrick Ion of Math Reviews for
helpful comments and encouragement.

BASIC INSTRUCTIONS

    After reading in the macro file using

preview or proof your figure with a coordinate grid printed on
top, by typing the following:

    \ShowGrid  
    \AffixLabels{<the graphics insertion>}

Here <the graphics insertion> is what you would type to insert
the graphics object alone without the grid.  This must provide
for the space around it. For example <the graphics insertion>
might well be \BoxedEPSF{MyFigure scaled 700} using the
boxedeps.tex macro package (from same source); this provides a
TeX box containing the encapsulated PostScript insert specified by
the file MyFigure. \AffixLabels{...} provides the grid (supposing
\ShowGrid is present) and later, once you have specified labels
using the grid, it will "tack on" the labels.

     The grid is a sort of (usually elongated) checkerboard of
ten rows and ten columns and its (internal) partitions are by
default numbered  .1, ... ,.9  both horizontally (X-coordinate
running left to right) and vertically (Y-coordinate running bottom
to top).  Thus the points enclosed by the grid correspond to the
points of the unit square in the cartesian "X-Y" plane, the lower
left corner corresponding to the origin (0,0).  By extrapolation,
the full page corresponds to a larger rectangle in the plane.

     These coordinates serve to position labels as follows.
Before the \AffixLabels{...} command type label specifications:

  \SetLabels
   (<X-coordinate>*<Y-coordinate>) <first label> \\
   .
   .
   .
   (<X-coordinate>*<Y-coordinate>)  <last label> \\
  \endSetLabels

Each row specifies one label and is terminated by \\.  In each
row, the position indicator comes first; it is written as a
standard cartesian point except that the X- and Y- coordinates
are separated by * rather than a comma because TeX allows a
comma as decimal point. There are no dimension units to specify
as the unit is the grid itself.

     By default, this cartesian point specifies where the middle
of the baseline of the label will be located.  However if you precede
the point by \L [or \R] the left [or right] edge of the baseline will
be located there. Similarly you may also precede the point by \T, \E,
or \B to vertically align the top equator or bottom of the label box
at the specified point.  This gives nine standard positions of
the label with respect to the insertion point --- corresponding to
the eight principle points of the compas and the center

                     \L\T     \T      \R\T

                     \L\E     \E      \R\E

                     \L\B     \B      \R\B

But this neglects the default "baseline" level of TeX,
giving potentially three more positions

                     \L    <no tag>   \R

For text, the baseline level is often the preferred. Its relation to
the others is variable. It will often coincide with the bottom level,
as happens for "X".  But it is often distinct, as for "g", in which
case you have in all 12 distinct positions rather than 9.

     It is convenient to think of this specification of label
position as attaching the label by a thumb-tack to the coordinate
grid. There are up to twelve positions of the thumb-tack on the
label, while the position of the thumb-tack on the coordinate grid is
arbitrary.  Normally, one choses the position of the thumb-tack on
the label to be the one that is the closest to the item being
labeled.  There are good reasons for this "rule of thumb":

   (a)  It facilitates correct positioning at first try.

   (b)  If the scale of the figure must be altered after labels
have been affixed, the labels have a good chance of remaining well
positioned.

   (c)  The visible grid need not extend beyond the "bounding box"
for the figure, because the best preferred position is always
(at least almost) within the bounding box .

The second reason is particularly important. Indeed it often
happens that scale has to be altered after labelling begins, in
order to either provide space for the labels, or to adjust
proportions between the labels and the figure.  (The size of labels
is unaffected by scaling.)

     Here is an artificial but self-contained test which uses
TeX rules to make a graphics object.

TEST

    Do not skip this!

 \def\FrameIt#1{\hbox{\vrule$\vcenter {\hrule\kern3pt%
             \hbox {\kern3pt #1\kern3pt}%
               \kern3pt\hrule}$\relax\vrule}}

 \def\Caption#1#2{\FrameIt{%
       \vtop {\hsize=#1\relax \parindent=0pt
         \leftskip=0pt \rightskip=0pt plus15pt
         \parfillskip=0pt
         \lineskip=1pt\baselineskip=0pt
         #2}}}

 \def\FirstQuadrant{\hbox to 100pt{\vrule\vbox to 100pt{%
        \hbox to 100pt{\hfil}\vfil\hrule}\hss}}

  \SetLabels
    \R(.5*.2) $\zeta\,\cdot$\\
    (.9*-.10) $\xi$\\
    \R(-.03*.9) $\eta$\\
    \T(.5*.9) \Caption{70pt}{%
          \it The norm of
          $g(\xi+i\eta)$ is indicated on
          contours of this invisible surface.}\\
  \endSetLabels

  \AffixLabels{\FirstQuadrant}

  \end

  Note that the coordinates to use for labels are indicated on the
edges of the grid (when visible) corresponding to the conventional
x- and y- axes of the Cartesian plane. By default the grid is
1-by-1. However, by the command \Edges{100}, you can change this
to 100-by-100 and many users find this alternative most
convenient. Place the command \Edges{...} in your style file (or
header) since its effect is is global. Other possible edge values
are 10 and 1000.

  If you use the command \Edges{...} at all, do so with care.  For
if you accidentally delete an \Edges{...} command your labels will
abruptly be badly misplaced and may logically but mysteriously
generate "dimension too big" errors under TeX and "off page" errors
under your driver.  

  You can dictate the edgescale for an individual figure by giving
the scale in brackets immediately after \SetLabels.  Thus, to
import into an article using say \Edge{100} a figure labelled using
another edgescale, say the original 1-by-1 default, you can use
\SetLabels[1]...\endSetLabels.

GETTING IT DOWN PAT

     Complicated labeling deserves the same respect as
complicated mathematics.  Do not expect it to come out perfect the
first time!  What is needed in either case is a mechanism to
repeatedly typeset troublesome pieces.

     One mechanism is always available.  One does complicated
labelling in a separate "test" file involving just the figure being
labelled;  a texpert will know how to \dump TeX's current state as
a temporary format that restarts rapidly at each retry.  Usually,
one then pastes the completed labelled figure back into the main
TeX file, but, of course, one can also \input it as an auxiliary
file.

     If you do not have a TeXpert at handy, here is a first
approximation to an efficient setup. By deletions reduce a copy
of your article to just a few lines before and after the figure.
Now label the figure, and finally, copy and paste the labelled
figure to the original article. Then copy the next figure to label
into this testbed and repeat. The TeXpert can improve the  speed
at which TeX starts up, by compiling a format specifically for
your article; just one caution: best NOT include in the format
ephemeral details of setup like \Set<mydriver>ArtSpecials (from
boxedeps.tex because this reads  figure dimensions which you may
change during your work session.

     An improved mechanism to repeatedly typeset troublesome
pieces is now available on the Macintosh; it is called LinoTeX;
see the same ftp sources.  It could be set up on many types
of computer.

     Before using labelfig.tex to attach labels to a graphics
object inserted using boxedeps.tex or BoxedArt.tex, make it a
firm rule to carefully adjust the bounding box using the trimming
commands of these packages, and also at least tentatively scale
and position the object. Beware of changing the grid inadvertently
after the labels have been positioned.  For example, correcting
the bounding box of a PostScript graphics object can foul up the
labels by changing the coordinate grid to which the labels are
attached. This is particularly true for the trimming  commands of
boxedeps.tex and BoxedArt.tex. However, as noted already, change
of scale is much less disruptive, and modest adjustments should be
well tolerated.

     Sometimes the labels protrude so far from the bounding box
of a figure that the figure has to be repositioned.  Best do this
by ad hoc spacing, say using \hglue and \vglue; altering the
bounding box would create a vicious circle.

     Remember that you are responsible for preventing labels
from overlapping. You are responsible for all label typography
including size and style. A label is really just about anything
that can be put in a TeX box. Note that spaces at the beginning
and end of labels will normally be suppressed; if you really want
them you must protect them with TeX braces.

     This package temporarily sets the \mathsurround parameter
of TeX to zero  while the labels are being affixed. This is done
because nonzero \mathsurround space would influence the position
of left and right aligned labels; then, when a texpert or printer
modifies mathsurround, diagram labeling might be disastrously
altered. There is a small price to pay involving labels that are
formatted as caption boxes including mathematics: you  may want or
need to specify an explicit mathsurround space within the caption
box; it will not influence anything outside.

     Those hostile to the use of * as separator between
the X and Y coordinates of label insertion points, are free to
impose another using \SetXYSeparator{<the new separator>}.  
Americans may prefer "," to "*" since they never use a 
comma as a decimal point; on the other hand, * may be more visible.

APPENDIX (I)  MERGING labelfig.tex LABELS INTO AN EPSF GRAPHICS OBJECT.

     As promised in the introduction, here is a recipe useful for
publishers. It works at least on Macintosh and at least for vectorized
graphics and Adobe type1 fonts.  (There is surely a similar recipe for
PCs under MSWindows.)

 (a)  Use boxedeps.tex utility to integrate the figure given by the eps
file, "x.eps" say, with a visible frame around it.  See
\ShowDisplacementBoxes command in boxedeps.tex.  To get precise results
automatically it is important to use the \Trim... commands of
boxedeps.tex making the "DisplacementBox" neatly fit the figure.

 (b)  Use the TeX printer driver and LaserWriter (versions >= 8.1.1) to
export to an EPSF the DVI page containing the integrated, labelled
figure. You now have an EPS file  "xx.eps"  that contains too much, and at
the wrong scale, and at wrong position.

 (c)  Convert the EPSF to an Adode Illustrator format EPSF using
the shareware utility called epsConvert by Sam Weiss
1993-- (currently $25).

 (d)  In Illustrator (or a compatible program), group the labels and the
"DisplacementBox"; copy them to the clipboard and paste them into "x.ps".
This step requires that all the label fonts be "visible to the Macintosh.

 (e)  Translate and scale the pasted group consisting of the labels plus
the "DisplacementBox" so as to make the "DisplacementBox" the bounding
box of (labelless) figure represented by "x.eps".  At this point the
labels will be correctly placed on the figure "x.eps".

 (f)  Ungroup and delete the "DisplacementBox".  The result is the
desired single EPS file, "x+.eps" say, It contains the original figure
plus its labels.  

     Using grouping and ungrouping appropriately in "x+.eps", a
publisher's staff can very efficiently improve label positions etc.

APPENDIX II)  SOME EXOTIC APPLICATIONS

     The grid of labelfig.tex is analogous to a light-table in
classical page makeup with wax or latex glue.  In principle, you
can use it to compose any page from its indivisible parts.  This
even has some of the artisanal charm of classical paste-up
provided you have a fast screen preview to make the process
"interactive".

     In practice labelfig.tex is a tool for nonstandard jobs.
Here are a few going beyond the labelling already discussed.

(I)  GRAPHICS INTEGRATION.

     This is accomplished by treating the imported graphics
objects as labels.  The underlying graphics object is then
typically an empty  \vbox to <dimension>{\vfill} in a TeX
\midinsert...\endinsert construction.  A label line
might be of the form

   (.1*.1) \\

The exact form of the special command varies from driver to
driver.  However, in the case of encapsulated PostScript graphics
(EPSF norm), by relying on boxedeps.tex, one can have the
following standard syntax (independant of driver  (see
boxedeps.doc for details.
  
  (.1*.1) \BoxedEPSF{MyFigure scaled <scale in mils>}\\

This may be slow since it requires TeX to read the PostScript
file to read bounding box using many complex macros.  So you
may want to try

  (.1*.1) \EPSFSpecial{MyFigure}{<scale in mils>}\\

which is fast and driver independant, but it squashes the
bounding box, normally to its lower left corner.

     Similarly for graphics of the Macintosh PICT norm ---
using BoxedArt.tex (same sources) in place of boxedeps.tex.

     This approach to integration is to be recommended when
one is assembling a composite graphics object.

 (II)  COMMUTATIVE DIAGRAM ENHANCEMENT

     Commutative diagrams or arrays of mathematical objects
connected by arrows of various sorts are common in mathematics.
The mathematical objects require the use of TeX.  Recently TeX
acquired a good collection of arrows of all slopes --- that of
LamSTeX --- plus pwerful macros to build the diagrams.

     However, even the LamSTeX collection is often
inadequate; it lacks for example double shafted arrows, dotted
arrows and curved arrows. Fortunately it is possible to produce
such arrows on an individual basis using sophisticated graphics
programs such as Illustrator and AldusFreehand (both serving
the EPSF norm) or using Metafont (with its public domain norm).
Since the creation of each new arrow is a work of love, you
probably want to limit the number of arrows by using LamSTeX
for most arrows. The 40K commutative diagram module of LamSTeX
has been adapted to work with AmSTeX and a copy may be posted
with LabelFig and related files. Unfortunately no one has yet
offered a version that works with Plain TeX or LaTeX.

       Suffice it here to say that when the exotic arrow has
been somehow imported into TeX, labelfig.tex treats it as a
label that one affixes to the commutative diagram.  Two other
steps will be treated in separate notes, namely the matter of
extracting the dimension specifications for the arrow and the
construction of the arrow --- for these steps are far from
unique and often depend intimately on your computer environment. 
Notes for the Macintosh-Textures-Illustrator combination are
found in the file ExoticArrows.doc.

 (III) NESTING 

Ingenuity pays off in exploiting labelfig.tex. One can
mix graphics and typography quite freely.  labelfig.tex is good
for freeform or overlapping arrangements, while boxedeps.tex (or
BoxedArt.tex) is best for regimented non-overlapping
arrangements --- and the two can be combined.

     The default behavior of labelfig.tex is not ideal 
for nesting objects, because to prevent trouble for beginners
the register for labels is globally cleared when \AffixLabels
concludes.  But there are switches available

      \LabelsGlobal      \LabelsLocal

which change this.  To understand this, extend the above test 
by something like:

 \LabelsLocal

 \SetLabels
    (.5*.5) AAA\\
 \endSetLabels

 {
 \SetLabels
    (.5*.5) ZZZ\\
 \endSetLabels
   \AffixLabels{\FirstQuadrant}
 }

   \AffixLabels{\FirstQuadrant}

     There are however potential pitfalls.  Neither
labelfig.tex nor boxedeps.tex has been tested under extreme
conditions. Problems may occur if their procedures are
indiscriminately nested. For boxedeps.tex (not labelfig.tex)
there is a precise cause for worry, namely many of its
variables are "global", which means that TeX braces will not
provide the protection one might expect.

COMMAND SUMMARY FOR labelfig.tex

  Here [...] means optional (one or zero)
       [...]* means any number of such constructs

  \SetLabels
    [[<P>](<X><Sep><Y>) <label> \\]*
  \endSetLabels
  \ShowGrid  
  \AffixLabels{<the figure>}

   --- <P> is tack position, one of eleven or empty
              order irrelevant

                   \L\T      \T      \R\T

                   \L\E      \E      \R\E

                     \L               \R

                   \L\B      \B      \R\B

   --- (<X><Sep><Y>) insertion point;
  <Sep> is separator, = * by default;
  \SetXYSeparator{<Sep>} changes it.
   <X> and <Y> are real numbers

  --- <label> a label to attach 

  --- <the figure> the figure to label 

  \GlobalLabels (default)     
  \LocalLabels  setting for nested constructs.

 \Grids makes ALL grids appear; \HideGrid then makes just next disappear.
 \noGrids returns to default.  The commands are always global.

 \GridLineWidth{<dimension>} adjusts width of grid lines. Default is very
small, to give "hairline" effect. If your grid lines are missing try
setting \GridLineWidth{1pt}.

 \Edges#1 globally changes the edge size of all grids to the numerical 
value #1, which must be 1, 10, 100, or 1000.  The default is 1.

VERSION HISTORY.
 --- Jan 1993: \Edges#1 and [??] option after \SetLabels
 --- July 1992: \Grids, \noGrids, \HideGrid;
       Gridlines become hairlines; \GridLineWidth{<dimension>}.
 --- Oct 1991, Jan 1992: \SetXYSeparator{<Sep>},  \LabelsGlobal,
       \LabelsLocal.
 --- July 1991: first release

Address for bugs and other feedback:

        Raymond S\'eroul
        IREM and Lab. de Typographie Informatise
        Univ. Rene Descartes
        Strasbourg

    Tel 33-88-41-63-45
    Email:  A18645@FRCCSC21.BITNET

        Laurent Siebenmann
        Mathematique, Bat. 425,
        Univ de Paris-Sud,
        91405-Orsay,
        France

    Tel 33-1-6941-7949; 
    Email: lcs@topo.math.u-psud.fr

\def\scalefig#1{\epsfxsize #1\textwidth}

\newcommand {\Ebb}{{\mathbb{E}}}

\newtheorem{theorem}{Theorem}

\newtheorem{lemma}{Lemma}

\newtheorem{corollary}{Corollary}

\title{\LARGE {Randomly-Directional Beamforming in Millimeter-Wave Multi-User  MISO Downlink}}

\author{
Gilwon Lee, {\em Student~Member, IEEE},  Youngchul
Sung$^\dagger$\thanks{$^\dagger$Corresponding author}, {\em
Senior~Member, IEEE}, and Junyeong Seo, {\em Student~Member, IEEE} \\
\thanks{The authors are with Dept. of Electrical Engineering,  KAIST, Daejeon 305-701, South
Korea. E-mail:\{gwlee@, ysung@ee., and jyseo@\}kaist.ac.kr.
This research was  supported by Basic Science Research Program through the National Research Foundation of Korea (NRF) funded by the Ministry of Education (2013R1A1A2A10060852). This research was supported by ’The Cross-Ministry Giga KOREA Project’ of The Ministry of Science, ICT and Future Planning, Korea. [GK14N0100, 5G mobile communication system development based on mmWave].}
}

\begin{document}

\maketitle

\begin{abstract}
In this paper, randomly-directional beamforming (RDB) is
considered for millimeter-wave (mm-wave) multi-user (MU)
multiple-input single-output (MISO) downlink systems. By using
asymptotic techniques, the performance of RDB and the MU gain in
mm-wave MISO are analyzed based on the uniform random
line-of-sight (UR-LoS) channel model suitable for highly
directional  mm-wave radio propagation channels. It is shown that
there exists a transition point on the number of users relative to
the number of antenna elements for non-trivial performance of the
RDB scheme, and furthermore sum rate scaling arbitrarily close to linear
scaling with respect to the number of antenna elements can be
achieved under the UR-LoS channel model by opportunistic random beamforming with proper user scheduling if the number of users
increases linearly with respect to the number of antenna elements.
The provided results yield insights into the most effective
beamforming and scheduling choices for mm-wave MU-MISO in various
operating conditions. Simulation results validate our analysis
based on asymptotic techniques for finite
 cases.
\end{abstract}


\section{Introduction}

{\em Motivation:} ~Recently, mm-wave multiple-input multiple-output (MIMO) operating in the band of
30-300GHz is considered as a promising technology to attain high
data rates for 5G wireless communications. Radio propagation in
the mm-wave band has several intrinsic properties; the propagation
in the mm-wave band is highly directional  with large path loss
and very few multi-paths. To compensate for the large path loss in
the mm-wave band, highly directional beamforming is required based
on large antenna arrays which can easily be implemented in the
mm-wave band due to small wavelength. To perform highly
directional downlink beamforming to a user in the cell, accurate
channel state information (CSI) is required at the base station
(BS). However, the channel is sparse in the arrival angle domain
and downlink channel estimation is difficult
\cite{Alkhateeb&Ayach&Leus&Heath:14JSTSP,Bajwa&Haupt&Sayeed&Nowak:10IEEE,Seo&Sung&Lee&Kim:15SPAWC}.
That is, it is difficult to identify the sparse propagation angle
and gain between the BS and an arbitrary receiver in the cell, and
identifying the sparse channel in the angle domain requires
sophisticated algorithms and heavy training overhead
\cite{Alkhateeb&Ayach&Leus&Heath:13ITA,Alkhateeb&Ayach&Leus&Heath:14JSTSP,Bajwa&Haupt&Sayeed&Nowak:10IEEE,Seo&Sung&Lee&Kim:15SPAWC}.
However, the focus of the existing channel estimation methods is
single-user mm-wave MIMO systems which do not have MU diversity.
Suppose directional downlink beamforming with a large uniform
linear array (ULA) of antenna elements at the BS. Although the
downlink beam is highly directional, it still has some beam width
because the number of antenna elements is finite in practice.
Thus, one might ask what happens if there are many users in the
cell and the BS just selects the transmission beam direction
randomly in the angle domain and looks for a receiver that happens
to be in the beam width of the selected beam of the BS. Of course,
if there exists only a single receiver in the cell, such {\em
randomly-directional beamforming (RDB)} with a narrow beam width
will not perform well because it will miss the receiver in most
cases. However,  if there exist more than one receivers randomly
located in the cell, the RDB scheme may perform reasonably well
with a sufficient number
 of users in the cell.
Then, a natural question is ``how many users in the cell are
enough for reasonable performance of such simple RDB and RDB with
multiple beams in the mm-wave band?'' In this paper, we
investigate the performance of RDB and the associated MU gain in
the mm-wave band to answer the above question.

{\em Channel model for mm-wave MIMO systems :} ~ Since the
performance of RDB depends on the channel model, answering the
above question should be based on a meaningful channel model. In
conventional lower band  MIMO communication, many MU gain analyses
were performed with the assumption of rich scattering, i.e.,
mostly under the independent and identically distributed (i.i.d.)
Rayleigh fading channel model or its variants such as correlated
fading or one-ring channel model
\cite{Sharif&Hassibi:05IT,Yoo&Goldsmith:06JSAC,TomasoniEtAl:09ISIT,AlNaffouri&Sharif&Hassibi:09COM,Hur&Tulino&Caire:12IT,Marzetta:10WC,Nam&Adhikary&Ahn&Caire:14JSTSP,Lee&Sung:14ITsub,Chung&Hwang&Kim&Kim:03JSAC,Viswanath&Tse&Laroia:02IT}.
However, the propagation in the mm-wave band is quite different
from that in the lower band; propagation in the mm-wave band is
highly directional and there are very few multi-paths in
propagation channels
\cite{Alkhateeb&Ayach&Leus&Heath:14JSTSP,Rappaport&BenDor&Murdock&Qiao:12ICC,Seo&Sung&Lee&Kim:15SPAWC,Alkhateeb&Ayach&Leus&Heath:13ITA}.
To model wireless channels in the mm-wave band,  {\em the UR-LoS channel model} was proposed in
\cite{Sayeed&Brady:13GC,Ngo&Larsson&Marzetta:14EUSIPCO}. The
UR-LoS channel model well captures the highly directional
propagation in the mm-wave band and is still analytically
tractable \cite{Sayeed&Brady:13GC,Ngo&Larsson&Marzetta:14EUSIPCO}.
Under the UR-LoS channel model,  the channel vector of each user
in the cell has a single LoS path component with a random
direction (or angle) and a random path gain. Since there is only
one  path in each user's channel under the UR-LoS channel model,
the UR-LoS channel model is a simplified channel model capturing
LoS propagation environments. To gain insights into random
beamforming in the mm-wave band and make performance analysis
tractable, we adopt the UR-LoS channel model in this paper even
though the actual channel may lie somewhere between the UR-LoS
channel model and the i.i.d. Rayleigh\footnote{Note that the
i.i.d. Rayleigh fading channel model for large antenna arrays is a
simplified model too. It is highly unlikely that each element of
the channel vector is i.i.d. when the channel vector size is very
large as in massive MIMO.} fading channel model.

{\em Summary of Results:} ~The MU gain under rich
scattering environments has been investigated extensively during
the last decade
\cite{Sharif&Hassibi:05IT,Yoo&Goldsmith:06JSAC,TomasoniEtAl:09ISIT,AlNaffouri&Sharif&Hassibi:09COM,Hur&Tulino&Caire:12IT,Marzetta:10WC,Nam&Adhikary&Ahn&Caire:14JSTSP,Lee&Sung:14ITsub,Chung&Hwang&Kim&Kim:03JSAC,Viswanath&Tse&Laroia:02IT}. However,  not much work has been done yet regarding  the MU gain
 in mm-wave MU-MISO/MIMO systems.
Recently, in \cite{Ngo&Larsson&Marzetta:14EUSIPCO}, Ngo {\it et al.}
 simplified the UR-LoS channel model as an urn-and-ball model  and numerically showed that
user scheduling can improve the worst-user performance. This work
provides an intuitive and insightful observation regarding the MU
gain in mm-wave MU-MISO, but the urn-and-ball channel model seems a
bit oversimplified compared to the UR-LoS channel model since the
urn-and-ball model does not consider non-orthogonal regions of
UR-LoS. (See Fig. \ref{fig:correlation}.)
In this paper, we rigorously analyze the RDB scheme, the associated
MU gain, and user scheduling  in mm-wave MU-MISO in an asymptotic regime in which the number of antenna elements tends to infinity, under the UR-LoS
channel model and the assumption
 of a ULA at the BS, and provide guidelines
for optimal operation in highly directional mm-wave MU-MISO
systems. The results of this paper are summarized in the below.

1) When $K=c_u M^q$ with $q\in(\frac{1}{2},1)$, where $K$ is
the number of users in the cell, $M$ is the number of antenna
elements, $q$ is the fraction order of $M$ for $K$, and $c_u$ is
some positive constant, the simple RDB scheme (in which the BS
transmits only one random beam, selects the user with the maximum
received signal power, and transmits to the selected user)
achieves $2q-1$ fraction of the rate performance with the
knowledge of perfect CSI as $M\to\infty$.  On the other hand, if
$K=c_u M^q$ with $q\in(0, \frac{1}{2})$, the simple RDB rate
converges to zero as $M\rightarrow \infty$. Hence, $K=c_u
\sqrt{M}$ is the transition point for the two distinct behaviors
of the RDB scheme.

2) When the BS sequentially transmits $S=c_b M^\ell$ beams
equi-spaced in the normalized angle domain with a uniform random
offset, selects the best beam among the $S$ beams that has the
maximum received power reported among all beams and all users, and
transmits  data with the best beam to the best user, this
multi-beam single-user RDB scheme achieves $2(q+\ell)-1$ fraction
of the optimal rate with perfect beamforming with perfect CSI as
$M\to \infty$, for $K=c_u M^q, S=c_b M^{\ell}$ ($q,\ell \in
(0,1)$), if $q+\ell \in (\frac{1}{2},1)$.

3) In the case of multi-beam and multiple-user selection RDB with
the UR-LoS channel model, sum rate scaling arbitrarily close to
linear scaling with respect to (w.r.t.) the number of antenna
elements can be achieved  by RDB with proper user scheduling. This
result is contrary to the existing result in rich scattering
environments that opportunistic random beamforming with user
selection does not provide a gain in the regime of a large number
of antennas under rich scattering environments
\cite{Sharif&Hassibi:05IT,TomasoniEtAl:09ISIT,Hur&Tulino&Caire:12IT}.

4) Combining the above results, we suggest optimal operation
for random beamforming in highly-directional mm-wave MISO
depending on the antenna array size and the number of users in the
cell, based on a newly defined metric named the fractional rate
order (FRO).

{\em Notations and Organization:}  ~~ Vectors and matrices are written
in boldface with matrices in capitals. For a matrix $\Abf$,
$\Abf^T$, $\Abf^H$, and $\mbox{tr}(\Abf)$ indicate the transpose,
conjugate  transpose, and trace  of $\Abf$, respectively. $\Ibf_n$
stands for the identity matrix of size $n$. (The subscript will be
omitted if unnecessary.)   The notation $\xbf\sim
\Cc\Nc(\mubf,\Sigmabf)$ means that $\xbf$ is complex Gaussian
distributed with mean vector $\mubf$ and covariance matrix
$\Sigmabf$, and $\theta \sim \mathrm{Unif}[a,b]$ means that
$\theta$ is uniformly distributed over the range $[a,b]$.
$\Ebb[\cdot]$ denotes the expectation. $|\Sc|$ denotes the
cardinality of $\Sc$. $\iota:=\sqrt{-1}$ and ${\mathbb{Z}}$ is the
set of integers. $a\uparrow b$ indicates that $a$ converges
to $b$ from the below.

The remainder of this paper is organized as follows. In Section
\ref{sec:systemmodel}, the system model and preliminaries are
described. In Section \ref{sec:randomdirecbeam}, the considered
RDB scheme is explained.  The asymptotic performance is analyzed
for the single beam case in Section \ref{sec:singlerandombeam} and
for the multiple beam case with single user selection or multiple
user selection in Section \ref{sec:multibeam}. Numerical results
are provided in Section \ref{sec:NumericalResult}, followed by
conclusions in Section \ref{sec:conclusion}.

\section{System Model and Preliminaries}
\label{sec:systemmodel}

We consider a single-cell mm-wave MU-MISO downlink system in which
a BS equipped with an  ULA of $M$ transmit antennas communicates
with $K$ single-antenna users. The received signal at user $k$ is
then given by
\begin{equation} \label{eq:received_signal}
y_k = \hbf_k^H\xbf + n_k, ~~~k=1,2,\cdots,K,
\end{equation}
where  $\hbf_k=[h_{k,1},h_{k,2},\cdots,h_{k,M}]^T$ is the channel
vector of user $k$, $\xbf$ is the transmitted signal vector
subject to a power constraint
$\mathrm{tr}(\mathbb{E}\{\xbf\xbf^H\}) \le P_t$, and $n_k \sim
\Cc\Nc(0,1)$  is the additive noise  at user $k$.

\subsection{Channel Model}
\label{subsec:ChannelModel}

For a typical mm-wave channel, there exist very few multipaths due
to the highly directional and quasi-optical nature of
electromagnetic wave propagation in the mm-wave band. In general,
a mm-wave channel is composed of a line-of-sight (LoS) propagation
component and a set of few single-bounce multipath components, and hence the mm-wave
channel for ULA systems can be modeled as
\cite{Rappaport&BenDor&Murdock&Qiao:12ICC}
\begin{equation}\label{eq:channel_vector_with_NLOS}
\hbf_k = \alpha_k \sqrt{M} \abf(\theta_k) + \sum_i\alpha_{k,i}
\sqrt{M}\abf(\theta_{k,i}), ~~ \text{for} ~~ k=1,\cdots,K,
\end{equation}
where $\alpha_k$ and $\theta_k$ are the complex gain and
 normalized direction of the LoS path for user $k$, $\{\alpha_{k,i}\}$ and
$\{\theta_{k,i}\}$ represent the complex gains and normalized directions of non-LoS (NLoS)
paths for user $k$, and $\abf(\theta)$ is the array steering vector given by
\begin{equation} \label{eq:array_steering_vec}
\abf(\theta) = \frac{1}{\sqrt{M}}[1,e^{-\iota \pi \theta},\cdots,e^{-\iota \pi(M-1)\theta}]^T.
\end{equation}
Here, the normalized direction $\theta$ is connected with the
physical angle of departure $\phi \in [-\pi/2,\pi/2]$ as $\theta =
\frac{2d\sin(\phi)}{\lambda}$, where $d$ and $\lambda$ are the
distance between two adjacent antenna elements and the carrier
wavelength, respectively. We assume the critically-sampled
environment, i.e., $\frac{d}{\lambda}=\frac{1}{2}$ in this paper.
Note that the array steering vector in
\eqref{eq:array_steering_vec} has unit norm and thus the
normalization factor $\sqrt{M}$ is included in
\eqref{eq:channel_vector_with_NLOS}.

For mm-wave channels with LoS links, the effect of NLoS
links is marginal since the path
loss of NLoS components is much larger than that of the LoS component;
 the power $|\alpha_{k,i}|^2$ associated with
NLoS paths is typically $20$dB weaker than the LoS component $|\alpha_{k}|^2$
\cite{Rappaport&BenDor&Murdock&Qiao:12ICC}. Hence,
we neglect the  NLoS components and consider the
LoS component only here, i.e., $\alpha_{k,i} = 0$ for $\forall i$
\cite{Sayeed&Brady:13GC,Bai&Desai&Heath:14ICNC}. We assume  that the LoS link gain is Gaussian-distributed, i.e., $\alpha_k \overset{\text{i.i.d.}}{\sim} \Cc\Nc(0,1)$  and that
 the normalized direction $\theta_k$ for each user $k$ is independent and identically
distributed (i.i.d.) with  $\theta_k \stackrel{i.i.d.}{\sim}\mathrm{Unif}[-1,1]$.
From the above assumptions, the mm-wave channel model \eqref{eq:channel_vector_with_NLOS} can be re-written as
\begin{equation} \label{eq:ChModelURLOS}
\hbf_k = \alpha_k \sqrt{M} \abf(\theta_k), ~~\text{for} ~~
k=1,\cdots,K.
\end{equation}
This channel model is the UR-LoS model considered in \cite{Ngo&Larsson&Marzetta:14EUSIPCO,Sayeed&Brady:13GC}. In this paper, we also adopt this channel model. Note that the power of  the UR-LoS channel model \eqref{eq:ChModelURLOS} is given by ${\mathbb{E}}\{||\hbf_k||^2\}=M$. Thus, the channel power linearly increases w.r.t. $M$ as in the i.i.d. Rayleigh channel model $\hbf_k \sim \Cc\Nc({\mathbf{0}},\Ibf)$. This means that the power radiated in the space is collected by the receiver antennas.

\subsection{Review of Opportunistic Random Beamforming in Rich Scattering Environments}

Before introducing the considered RDB  for large mm-wave MIMO
systems with the UR-LoS channel model, we briefly review the
random (orthogonal) beamforming (RBF) scheme in
\cite{Sharif&Hassibi:05IT} devised  for rich scattering
environments under which each element $h_{k,j}$ in the channel
vector $\hbf_k$ has an i.i.d. Rayleigh fading:
\begin{equation}  \label{eq:richscatteringgkj}
h_{k,j} \stackrel{i.i.d.}{\sim} \Cc\Nc(0,\sigma_h^2) ~\text{for}
~j=1,\cdots,M.
\end{equation}
 In the RBF scheme, the BS
constructs a set of $S$ random orthonormal beam vectors $\{\ubf_1,\cdots,\ubf_S\}$
and transmits each beam sequentially to the $K$ users in the cell
during the training period. Then, each user $k$
computes the signal-to-interference-plus-noise ratio (SINR) for each beam direction at the end of the training period, given by $\mathrm{SINR}_{k,i} = \frac{\frac{P_t}{S}|\hbf_k^H\ubf_i|^2}{
1+\frac{P_t}{S}\sum_{j\neq i}|\hbf_k^H\ubf_j|^2}$ for $i=1,\cdots,S$.
After the training period, each user $k$ feeds back its maximum
SINR value, i.e., $\max_{1\le i\le S} \mathrm{SINR}_{k,i}$, and
the beam index $i$ at which the SINR is maximum. Then, after the
feedback the BS assigns each beam $i$ to the user $k^\prime(i)$
with the highest SINR for  beam $i$, i.e., $k^\prime(i) =
\mathop{\arg\max}_{1\le k \le K}\mathrm{SINR}_{k,i}$, and transmits $S$ data streams to the selected $S$ users. In
\cite{Sharif&Hassibi:05IT}, Sharif and Hassibi derived several
scaling laws of this RBF scheme in the case of $S=M$ with the {\em
small-scale}\footnote{In small-scale MIMO systems, $M$ is small
and $K$ is relatively large. Hence, the authors of
\cite{Sharif&Hassibi:05IT} focused on the asymptotic scenario in
which $K$ grows to infinity with fixed $M$ or  $M$ growing much
slower than $K$. Note that $K=\Theta(e^M)$ for $K$ as a function
of $M$ for the scaling of $M=\Theta(\log K)$ considered in
\cite{Sharif&Hassibi:05IT}.} MIMO in mind, i.e., $M \ll K$, as $K
\rightarrow \infty$. Specifically, they showed
\begin{equation}
\Rc_{RBF} \sim_K \left\{
\begin{array}{lll}
M\log \log K, &\text{as}~ K \rightarrow \infty, &\text{for fixed}~ M, \\
cM,           &\text{as}~ K \rightarrow \infty, &\text{for} ~ M=O(\log K),
\end{array}
\right.
\end{equation}
where $\Rc_{RBF}=\mathbb{E}\left[\sum_{i=1}^M\log\left(1+\max_{1\le k \le K} \mathrm{SINR}_{k,i}\right)\right]$
and $c$ is a positive constant. (Here, $x\sim_K y$ indicates that $\lim_{K\to \infty}x/y=1$.)
 Furthermore, they showed that \cite{Sharif&Hassibi:05IT}
\begin{equation}\label{eq:RBF_sumrate}
\lim_{K \to \infty}\frac{\Rc_{RBF}}{M} = 0,
\end{equation}
if $\lim_{K \to \infty} \frac{M}{\log K} = \infty$ (here, $\lim_{K \to \infty} \frac{M}{\log K} = \infty$
 is equivalent to $\lim_{K \to \infty} \frac{\log K}{M} = 0$).
The above scaling laws state that the sum rate of the RBF scheme
maintains linear scaling w.r.t.  the number $M$ of transmit
antennas when $M$ grows no faster than $\log K$ as $K \rightarrow
\infty$, but this linear scaling with $M$ is not achieved when $M$
grows faster than $\log K$ as $K \rightarrow \infty$. That is, the
RBF scheme performs well, i.e., the RBF data rate grows linearly
w.r.t. the number $M$ of antennas in small-scale MIMO systems with
a large number of users in the cell, but does not show linear
scaling rate w.r.t. $M$ in massive MIMO situations under rich
scattering environments.

Now consider the case of  mm-wave MIMO with the UR-LoS channel
model. Due to large path loss in the mm-wave band, highly
directional beamforming is required to compensate for the large
path loss. This means a large antenna array at the BS, i.e., $M$
is very large. In the following sections, we investigate the
performance of random beamforming under the UR-LoS channel model
in a progressive manner from one single random beam and single
user selection to multiple random (asymptotically-orthogonal)
beams and multiple user selection under a massive MIMO asymptote
in which $M$ goes to infinity. Note that under the UR-LoS channel
model the randomness in beams lies in the beam direction. Thus,
random beamforming under the UR-LoS channel model is named
randomly-directional beamforming (RDB) in this paper.

\section{Randomly-Directional Beamforming in Massive mm-Wave MISO }
\label{sec:randomdirecbeam}

First consider the RDB strategy
in the single beam downlink transmission case. In this case,
during the training period, the BS chooses a normalized direction
$\vartheta$ randomly and transmits the beam $\xbf$ in
\eqref{eq:received_signal}  given by
\begin{equation}  \label{eq:xbfFirstSecIII}
\xbf = \abf(\vartheta)
\end{equation}
where $\vartheta \sim \mathrm{Unif}[-1,1]$ and $\abf(\theta)$ is
given by \eqref{eq:array_steering_vec}. (We simply set $P_t=1$ for
simplicity here.) Then, each user $k$ in the cell composed of $K$
users feeds back the average received power\footnote{To average
out the noise effect, each user can have multiple time samples
$y_k(i)$ during the training period and average the multiple
samples  for the feedback value
$|\bar{y}_k|^2=|\frac{1}{N_s}\sum_{i=1}^{N_s} y_k(i)|^2
\stackrel{(a)}{=}|\hbf_k^H\xbf|^2+\frac{1}{N_s}$. We assume that
sufficient sample average is done and will ignore possible error
in step (a) in this paper.} $|\bar{y}_k|^2 ~(\approx |\hbf_k
\xbf|^2 + \frac{1}{N_s})$ to the BS, where
$|\hbf_k^H\xbf|^2=|\alpha_k|^2 \cdot
M|\abf(\theta_k)^H\abf(\vartheta)|^2$. After the feedback period
is over, the BS selects the user that has maximum signal power and
transmits a data stream with the beamforming vector $\xbf$ in
\eqref{eq:xbfFirstSecIII} to the user. Then,
 the expected  rate $\Rc_1$ of the RDB scheme  is given by
\begin{equation} \label{eq:expected_rate_s1}
\Rc_1=\mathbb{E}\left[
\log\left(1+ \max_{1\le k \le K}
|\alpha_k|^2 M|\abf(\theta_k)^H\abf(\vartheta)|^2 \right)
\right],
\end{equation}
where the expectation is over $\hbf_k$ and $\xbf$. Consider
the case of $K=1$. In this case, we have an upper bound on
$\Rc_1$ from Jensen's inequality  as
\begin{align}
\Rc_1 &= \mathbb{E}\left[\log\left(1+|\alpha_1|^2M|\abf(\theta_1)^H\abf(\vartheta)|^2\right)\right] \le \log\left(1+\mathbb{E}\left[|\alpha_1|^2 M |\abf(\theta_1)^H \abf(\vartheta)|^2\right]\right) \nonumber \\
&= \log\left(1+\mathbb{E}\left[|\alpha_1|^2\right]
\mathbb{E}\left[M |\abf(\theta_1)^H
\abf(\vartheta)|^2\right]\right) = \log 2.
\end{align}
(It will be shown in the next section that $\Rc_1$ actually goes to zero as $M\to \infty$.)
The last equality holds from
$\mathbb{E}[|\alpha_1|^2]=1$ because $|\alpha_1|^2$ has a chi-square distribution with degree-of-freedom two, i.e.,  $|\alpha_1|^2 \sim \chi^2(2)$ and from
\[
\mathbb{E}[M |\abf(\theta_1)^H \abf(\vartheta)|^2]
=\frac{1}{M}\mathbb{E}\left[\left|\sum_{n=0}^{M-1} e^{-\iota\pi n (\vartheta-\theta_1)}\right|^2\right]
= \frac{1}{M} \mathbb{E}\left[M + \sum_{\substack{n,m \\ n \neq m}}e^{-\iota\pi(m-n)(\vartheta - \theta_1)}\right]
\overset{(a)}{=}1,
\]
where step $(a)$ holds because
$\mathbb{E}[e^{-j\pi(m-n)(\vartheta-\theta_1)}] =
\frac{1}{2}\int_{-1}^1 e^{-j \pi
(m-n)\tilde{\theta}_k}d\tilde{\theta}_k = \frac{\sin
\pi(m-n)}{\pi(m-n)}=0$ for any $(m-n) \in
\mathbb{Z}\backslash\{0\}$\cite{Ngo&Larsson&Marzetta:14EUSIPCO}.
\footnote{We can regard $\tilde{\theta}_k :=\vartheta-\theta_k
\sim \mathrm{Unif}[-1,1]$ in case that $\vartheta-\theta_k$
appears as $e^{\iota\pi l(\vartheta - \theta_k)}$ for any integer
$l$ due to the periodicity of period two. See Appendix A.}
\begin{figure}[t]
\centering
\includegraphics[width=3.5in]{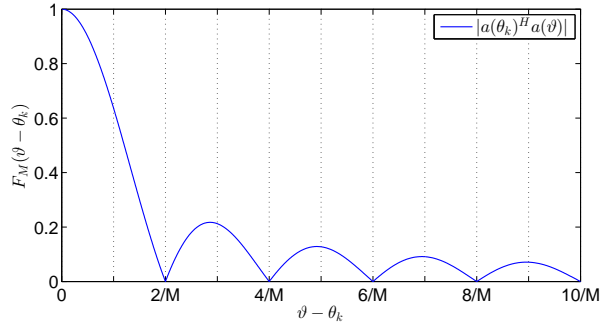}
\caption{$F_M(\vartheta-\theta_k)$ in \eqref{eq:beam_patterngain}  when $M=100$. }
\label{fig:correlation}
\end{figure}
Thus, the rate of the RDB scheme for $K=1$ is insignificant
regardless of the value of $M$. In this case, it is imperative to
obtain the CSI of the single user to achieve the attainable rate
of $\log(1+|\alpha_1|^2 M)\sim_M \log M$
\cite{Alkhateeb&Ayach&Leus&Heath:14JSTSP,Bajwa&Haupt&Sayeed&Nowak:10IEEE,Seo&Sung&Lee&Kim:15SPAWC}.
However, the situation becomes different as $K$ becomes large. In
order to obtain an insight
 into the MU gain in the RDB scheme in mm-wave massive MIMO with the UR-LoS channel model before rigorous analysis in the next section, let us examine the relationship
between $|\abf(\theta_k)^H\abf(\vartheta)|$ and $\{\theta_k, \vartheta\}$:
\begin{align}
|\abf(\theta_k)^H\abf(\vartheta)|
&= \frac{1}{M}\left|
\sum_{n=0}^{M-1} e^{-\iota\pi n (\vartheta - \theta_k)}
\right|
= \frac{1}{M}\left|
\frac{1-e^{-\iota\pi(\vartheta - \theta_k)M}}{1-e^{-\iota\pi(\vartheta - \theta_k)}}
\right| \nonumber\\
&= \frac{1}{M} \left| \frac{\sin \frac{\pi(\vartheta -
\theta_k)M}{2}}{\sin\frac{\pi(\vartheta-\theta_k)}{2}}
\right|=:F_M(\vartheta - \theta_k), \label{eq:beam_patterngain}
\end{align}
which is the  Fej\'er kernel $F_M(\cdot)$ of order $M$
\cite{Strichartz:book}. Fig. \ref{fig:correlation} shows the value
of \eqref{eq:beam_patterngain} versus $\vartheta-\theta_k$. From
\eqref{eq:beam_patterngain}, we have
$|\abf(\theta_k)^H\abf(\vartheta)| \to 0$ as $M \to \infty$ for
fixed $\vartheta$ and $\theta_k$. On the other hand, we have $
|\abf(\theta_k)^H\abf(\vartheta)| \to \left|\frac{2\sin\frac{\pi
\Delta}{2}}{\pi \Delta}\right|$ as $M\to\infty$, provided that
$\vartheta- \theta_k = \frac{\Delta}{M}$ for some $\Delta > 0$
\cite{Ngo&Larsson&Marzetta:14EUSIPCO}. This is because
\begin{equation}
\frac{1}{M} \left| \frac{\sin \frac{\pi(\vartheta -
\theta_k)M}{2}}{\sin\frac{\pi(\vartheta-\theta_k)}{2}}\right|
\overset{(a)}{\approx} \frac{1}{M} \left| \frac{\sin \frac{\pi
\Delta}{2}}{\frac{\pi \Delta}{2M}}\right| \to \left|
\frac{2\sin\frac{\pi \Delta}{2}}{\pi \Delta} \right|
\end{equation}
where $(a)$ holds from $\sin \epsilon \approx \epsilon$ for small
$\epsilon>0$. That is, the asymptotic value of
$|\abf(\theta_k)^H\abf(\vartheta)|$ may not be zero if
$(\vartheta-\theta_k)$ becomes sufficiently small in the order of
$O\left( \frac{1}{M}\right)$ as $M \to \infty$.   On the other
hand, one can show in a similar way that
$|\abf(\theta_k)^H\abf(\vartheta)| \to 0$ as $M\to\infty$,  when
$\vartheta-\theta_k = \frac{\Delta}{M^\alpha}$ for some $\alpha
<1$ and $\Delta$.

Now, suppose that we can find a user $k$ such that $|\vartheta -
\theta_k|<\frac{1}{M}$ almost surely due to  MU diversity. Then,
the rate $\Rc_1$ of the RDB scheme is lower bounded by
\begin{equation}
\Rc_1  \ge \mathbb{E}\left[\log\left(1+|\alpha_k|^2 M \frac{4}{\pi^2}\right)\right]
\sim_M \log M,
\end{equation}
as $M\rightarrow \infty$. In other words, if the number $K$ of
users as a function of $M$ is sufficiently large such that there
exists a user $k$ for whom $|\vartheta-\theta_k|$ is sufficiently
small in the order of $O\left( \frac{1}{M}\right)$ with high
probability, the RDB scheme has  asymptotically good performance.

\section{Asymptotic Analysis of The RDB Rate: The Single Beam  Case}
\label{sec:singlerandombeam}

In this section, we rigorously analyze the asymptotic performance
of the RDB scheme in the single downlink beam case. Direct
computation of $\Rc_1$ in  \eqref{eq:expected_rate_s1} is
difficult since the integral in \eqref{eq:expected_rate_s1} does
not have a closed-form expression. To circumvent this difficulty,
we use several techniques to bound $\Rc_1$  by first assuming that
$\alpha_k = 1$ for all $k$ and focusing on the term
$Z_k:=M|\abf(\theta_k)^H\abf(\vartheta)|^2$ in
$\eqref{eq:expected_rate_s1}$. Then, we will include the term
$\alpha_k \sim \Cc\Nc(0,1)$ in the performance analysis later. We
begin with the following lemma.

\vspace{0.5em}

\begin{lemma}\label{lem:lemma1_bounds_of_Zk}
For any constant $p \in (-1,1)$ and sufficiently large $M$,
we have
\begin{equation}\label{eq:lemma1_bounds_of_tildetheta}
|\tilde{\theta}_k|<\frac{1}{\frac{\pi}{4}M^{(1+p)/2}}
\end{equation}
under the event $\{Z_k > M^p\}$, and furthermore
\begin{equation} \label{eq:lemma1_bounds_of_Zk}
\frac{1}{2\pi M^{(1+p)/2}} < \mathrm{Pr}\{Z_k > M^p\} <
\frac{1}{\frac{\pi}{4}M^{(1+p)/2}},
\end{equation}
 where $\tilde{\theta}_k
=\vartheta-\theta_k$ and $Z_k =
M|\abf(\theta_k)^H\abf(\vartheta)|^2$.
\end{lemma}

\begin{proof}
From \eqref{eq:beam_patterngain}, the event $\{ Z_k=M|\abf(\theta_k)^H\abf(\vartheta)|^2 > M^p \}$ is equivalent to
\begin{equation}
\left|\frac{\sin\frac{\pi{\tilde{\theta}}_k M}{2}}{\sin
\frac{\pi {\tilde{\theta}}_k}{2}}\right|
>  M^{(1+p)/2}, \label{eq:lemma1_condition}
\end{equation}
where ${\tilde{\theta}}_k \sim \mathrm{Unif}[-1,1]$ by Appendix \ref{append:uniformdist}.
A necessary condition to satisfy \eqref{eq:lemma1_condition} is that the denominator in the left-hand side (LHS) of
\eqref{eq:lemma1_condition} should be upper bounded as
\begin{equation}\label{eq:lemma1_sineq}
\left|\sin\frac{\pi {\tilde{\theta}}_k}{2}\right| < \frac{1}{M^{(1+p)/2}}
\end{equation}
since the numerator $\left|\sin \frac{\pi {\tilde{\theta}}_k
M}{2}\right| \le 1$ and $M^{(1+p)/2} > 1$ for $p \in (-1,1)$ and $M >1$. For
given $p \in (-1,1)$, the upper bound in the right-hand side (RHS)
of \eqref{eq:lemma1_sineq} goes to zero as $M\to \infty$. Hence,
by the fact that $ \frac{\epsilon}{2} < \sin \epsilon$ for small
$\epsilon>0$, \eqref{eq:lemma1_sineq} implies
\begin{equation}\label{eq:lemmma1_mukupperbound}
| {\tilde{\theta}}_k| < \frac{1}{\frac{\pi}{4}M^{(1+p)/2}}
\end{equation}
for sufficiently large $M$. Therefore,
\eqref{eq:lemma1_bounds_of_tildetheta} holds and we have the upper
bound in \eqref{eq:lemma1_bounds_of_Zk}, since
\eqref{eq:lemmma1_mukupperbound} is a necessary condition for
$\{Z_k > M^p\}$:
\begin{equation}
\mathrm{Pr}\{Z_k > M^p\} < \mathrm{Pr}\left\{ |
{\tilde{\theta}}_k| < \frac{1}{\frac{\pi}{4}M^{(1+p)/2}}
\right\}=\frac{1}{\frac{\pi}{4}M^{(1+p)/2}}
\end{equation}
for sufficiently large $M$, since $\tilde{\theta}_k \sim \mbox{Unif}[-1,1]$.

Now consider the lower bound in \eqref{eq:lemma1_bounds_of_Zk}.
From the fact that $\sin \epsilon < \epsilon$ for $\epsilon>0$, we
have
\begin{equation}  \label{eq:SecIVlem1LB1}
\left|\sin\frac{\pi {\tilde{\theta}}_k}{2}\right| < \frac{1}{2M^{(1+p)/2}}
\end{equation}
if
\begin{equation} \label{eq:SecIVlem1LB11}
\frac{\pi}{2}|{\tilde{\theta}}_k| < \frac{1}{2 M^{(1+p)/2}}.
\end{equation}
If the following equation
\begin{equation} \label{eq:lemma1_sin_pi_muk_m}
\left|\sin\frac{\pi {\tilde{\theta}}_k M}{2}\right| \ge \frac{1}{2}
\end{equation}
is satisfied in addition to \eqref{eq:SecIVlem1LB11} implying
\eqref{eq:SecIVlem1LB1}, then \eqref{eq:lemma1_condition} is
satisfied (i.e., the joint event of \eqref{eq:SecIVlem1LB11} and
\eqref{eq:lemma1_sin_pi_muk_m} is a sufficient condition for
\eqref{eq:lemma1_condition}). It is easy to see that the solution
to \eqref{eq:lemma1_sin_pi_muk_m} is
\begin{equation} \label{eq:condition_of_mu_k}
|{\tilde{\theta}}_k| \in \left\{\left[\frac{2k}{M}+\frac{1}{3M}, \frac{2k}{M}+\frac{5}{3M}  \right], k=0,1,2,\cdots,\right\}.
\end{equation}
Note that  $\left|\sin\frac{\pi {\tilde{\theta}}_k M}{2}\right|$
in \eqref{eq:lemma1_sin_pi_muk_m} has period $\frac{2}{M}$ and the
length of one interval per period contained in the set
\eqref{eq:condition_of_mu_k} is $\frac{4}{3M}$. Hence, the set
\eqref{eq:condition_of_mu_k} occupies $\frac{2}{3}$ length of each
period of $\frac{2}{M}$. Since the term $\frac{1}{M^{(p+1)/2}}$
for given $p \in (-1,1)$ converges to zero  slower than
$\frac{1}{M}$ as $M \rightarrow \infty$, multiple discontinuous
intervals in the set \eqref{eq:condition_of_mu_k} are contained in
the set defined by \eqref{eq:SecIVlem1LB11}, and the length of the
intersection of  the sets \eqref{eq:SecIVlem1LB11} and
\eqref{eq:condition_of_mu_k} is lower bounded by
$\frac{2}{3}\left(\frac{1}{\pi M^{(1+p)/2}} - \frac{2}{M}\right)$,
where minus $\frac{2}{M}$ takes into account the impact of the
last possibly partially overlapping interval. Hence, we have the
lower bound part of  \eqref{eq:lemma1_bounds_of_Zk}:
\begin{align}
\mathrm{Pr}\{Z_k >  M^p\} &\ge \frac{2}{3}\left(\frac{1}{\pi
M^{(1+p)/2}} - \frac{2}{M}\right) > \frac{1}{2}\cdot\frac{1}{\pi
M^{(1+p)/2}}.
\end{align}
 for
sufficiently large $M$. $\hfill \square$
\end{proof}

\vspace{0.5em}

Using Lemma \ref{lem:lemma1_bounds_of_Zk} we have the following
theorem.

\begin{theorem} \label{thm:thm1_asymp_bounds}
For $K= M^q$ and $q \in (0,1)$, we have asymptotic upper and lower bounds for
 $\Rc_1$ in \eqref{eq:expected_rate_s1} when $\alpha_k=1$ for all $k$, given by
\begin{equation} \label{eq:thm1_asymp_bounds}
\log(1+M^{2q-1-\epsilon}) ~\lesssim_M~ \mathbb{E}\left[\log(1+Z)\right] ~\lesssim_M~
\log(1+M^{2q-1+\epsilon})
\end{equation}
for any sufficiently small $\epsilon >0$, where $Z = \underset{1\le k\le K}{\max} Z_k$ and $x\lesssim_M y$
means
$\lim_{M\to\infty} x/y \le 1$.
\end{theorem}

\begin{proof}
The probability of the event $\{Z > M^p\}$ for any $p \in (-1,1)$
can be expressed as
\begin{align}
\mathrm{Pr}\left\{ \max_k Z_k > M^p\right\}
&= 1 - \mathrm{Pr}\{Z_k \le M^p\}^K \label{eq:thm1_cdf_zk}\\
&= 1 - \left( 1 - \frac{1}{c_MM^{(1+p)/2}}\right)^K ,
\label{eq:thm1_proof_eq}
\end{align}
where the second equality holds by \eqref{eq:lemma1_bounds_of_Zk}
of Lemma 1 ($c_M$ is bounded between $\frac{\pi}{4}$ and $2 \pi$
for all sufficiently large $M$). We consider the second term in
\eqref{eq:thm1_proof_eq}. Pick $p=2q-1 -\epsilon$ for small
$\epsilon > 0$ such that $p \in (-1,1)$. (For such $\epsilon$,
$2q-\epsilon >0$.) Then the second term can be expressed as
\begin{align}
\left( 1 - \frac{1}{c_M M^{(1+p)/2}} \right)^K
&= \left( 1 - \frac{1}{c_M M^{q-\frac{\epsilon}{2}}} \right)^{M^q} \label{eq:thm1_proof_second}\\
&= e^{M^q \log \left( 1- \frac{1}{c_M M^{q-\frac{\epsilon}{2}}} \right)} \\
&= e^{-\frac{1}{c_M}M^{\epsilon/2} + O\left(
\frac{1}{M^{2q-\epsilon}} \right)} ~\to 0 ~~\text{as}~~ M\to
\infty, \label{eq:thm1_proof_third}
\end{align}
where we used the fact that $\log(1-x) = -x + O(x^2)$ for small $x$
in the third step. Therefore, in this case, we have
\begin{equation}\label{eq:asymp_prop_Z_k}
\mathrm{Pr}\{Z>M^p\}   \to 1, ~~\text{as}~~ M \to\infty
\end{equation} and thus
$\mathbb{E}[\log(1+Z)]$ can be bounded as
\begin{align}
\mathbb{E}[\log(1+Z)]
&\ge \int_{M^p}^M \log(1+z)p(z) dz \\
&\ge \log(1+M^p) \int_{M^p}^M p(z) dz \\
&\sim_M \log(1+M^p), ~~\text{as}~~ M\to\infty,
\end{align}
since $\int_{M^p}^M p(z) dz = \mbox{Pr}\{Z > M^p\} \rightarrow 1$ in this case.
Hence, the claim on the lower bound follows.

Now pick $p=2q-1+\epsilon$ for small $\epsilon > 0 $ such that $p \in (-1,1)$. (For such $\epsilon$, $2q+\epsilon > 0$.)
Then, by the techniques used in \eqref{eq:thm1_proof_second}-\eqref{eq:thm1_proof_third}, the second term in \eqref{eq:thm1_proof_eq} can be computed as
\begin{align}
\left(1-\frac{1}{c_M M^{(1+p)/2}}\right)^K &=
e^{-\frac{1}{c_M}M^{-\epsilon/2}
+ O\left(\frac{1}{M^{2q+\epsilon}}\right)} \\
&\stackrel{(a)}{=} 1 + O\left( -\frac{1}{M^{\frac{\epsilon}{2}}}+\frac{1}{M^{2q+\epsilon}}\right) = 1 - O\left( \frac{1}{M^{\frac{\epsilon}{2}}}\right) ,
\end{align}
where the step (a)  holds by the identity $e^x=1+O(x)$ for small $x$.
Therefore, in this case, the probability of the event $\{Z > M^p\}$ is given by
$O\left(\frac{1}{M^{\epsilon/2}}\right)$. Using
this, we have
\begin{align}
\mathbb{E}[\log(1+Z)]
&= \int_{M^p}^M \log(1+z)p(z)dz
+ \int_0^{M^p} \log (1+z)p(z) dz \\
&\le \log(1+M) O\left(\frac{1}{M^{\epsilon/2}}\right)
+ \log(1+M^p) \\
&\sim_M \log(1+M^p), ~~\text{as}~~ M \to \infty.
\end{align}
In the second step, we used $\int_{M^p}^M p(z)dz=\mbox{Pr}\{ Z > M^p\} =O\left(\frac{1}{M^{\epsilon/2}}\right)$.  Hence, the claim on the upper bound follows. $\hfill \square$
\end{proof}

\vspace{0.5em}

Theorem \ref{thm:thm1_asymp_bounds} states that the single-beam
RDB scheme under the assumption $\alpha_k=1,~\forall~k$ has
asymptotically nontrivial performance, i.e., $\Rc_1 \to \infty$,
as $M\rightarrow \infty$, when $K= M^{q}$ with $q \in
(\frac{1}{2},1)$. On the other hand, when $K= M^{q}$ with $q \in
(0,\frac{1}{2})$, the RDB scheme has trivial performance, i.e.,
$\Rc_1 \to 0$, as $M \rightarrow \infty$. Thus, $q=\frac{1}{2}$ is
the performance transition point for the single-beam RDB scheme
under the UR-LoS channel model.

Now consider the impact of the path gain term $\alpha_k
\overset{i.i.d.}{\sim} \Cc\Nc(0,1)$ on the single-beam RDB rate
$\Rc_1$. In fact, the same is true under the assumption of $\alpha_k
\overset{i.i.d.}{\sim} \Cc\Nc(0,1)$.

\vspace{0.5em}

\begin{theorem}\label{thm3:ratio_of_Rto_PerfCSI}
For $K = M^q$ with $q \in (\frac{1}{2},1)$ and $\alpha_k
\stackrel{i.i.d.}{\sim} \Cc\Nc(0,1)$, we have
\begin{equation} \label{eq:ratio_of_Rto_PerfCSI}
\lim_{M \to \infty}\frac{\Rc_1}{\mathbb{E}\left[\log(1+M\max_k|\alpha_{k}|^2)\right]} = 2q-1,
\end{equation}
where $\Rc_1$ is the optimal single-beam RDB rate defined in
\eqref{eq:expected_rate_s1} considering the random path gain, and
$\mathbb{E}\left[\log(1+M\max_k|\alpha_{k}|^2)\right]$ is the
optimal rate of exact beamforming based on  perfect CSI at the BS.
On the other hand, when $q \in  (0, \frac{1}{2})$, $\Rc_1
\rightarrow 0$ as $M \rightarrow \infty$.
\end{theorem}

\begin{proof}
See Appendix \ref{appen:thm3}. $\hfill \square$
\end{proof}

\vspace{0.5em}

\noindent Note that the ratio $2q-1$ of the RDB rate $\Rc_1$ to the exact beamforming rate is the same for both assumptions $\alpha_k=1$ and $\alpha_k \sim \Cc\Nc(0,1)$.  As seen, the single-beam RDB
strategy achieves
$2q-1$ fraction of the exact beamforming rate based on perfect CSI
at the BS.   The supremum fraction of one can be achieved
arbitrarily closely when the number $K$ of users grows almost
linearly w.r.t. $M$, i.e., $q$ is arbitrarily close to one.

Theorems \ref{thm:thm1_asymp_bounds} and
\ref{thm3:ratio_of_Rto_PerfCSI} hold exactly in the same form when
$K=c_u M^q$ for any constant $c_u > 0$. However, for the
notational simplicity in the proofs, we just used $K=M^q$. Note
that in the single beam case we only have the power gain by the
antenna array, as shown in the maximum rate of
$\mathbb{E}\left[\log(1+M\max_k|\alpha_{k}|^2)\right]$ even by
perfect beamforming.

\section{Asymptotic Analysis of The RDB Rate: The Multiple Beam Case}
\label{sec:multibeam}

In this section, we consider the case in which the number $S$ of
randomly-directional beams is more than one and   allowed to grow
to infinity as a function of $M$, and analyze the corresponding
asymptotic performance. In the multiple beam case, the BS
transmits $S$ random beams equi-spaced in the normalized angle
domain, defined as
\begin{equation}\label{eq:equispaced_beam_def}
\wbf_b = \abf(\vartheta_b) = \abf\left(\vartheta +
\frac{2(b-1)}{S}\right), ~~ \text{for} ~~ b=1,\cdots,S,
\end{equation}
where $\vartheta \sim \mathrm{Unif}[-1,1]$, to the downlink
sequentially during the training period. We assume that the
network is synchronized and thus each user knows the training beam
index $b$ by the corresponding training interval.  Here, the
difference between the normalized directions of two adjacent beams
is $\frac{2}{S}$ and the offset $\vartheta$ is randomly generated
on $(-1,1]$. (Recall from \eqref{eq:array_steering_vec} that
$\abf(\theta)$ is periodic in $\theta$ with period 2.)  Note that
the equi-spaced beams are asymptotically orthogonal to one
another, i.e.,
\begin{equation}  \label{eq:multibeamortho}
\lim_{M \to \infty}|\abf(\vartheta_{b_1})^H\abf(\vartheta_{b_2})|
= 0 ~~\text{for}~~ b_1 \neq b_2,
\end{equation}
when $S=o(M)$.

In the next subsections, we analyze the asymptotic performance of
single user selection based on multiple training beams first and
multiple user selection  based on multiple beams later.

\subsection{The Single User Selection Case}

In the single user selection case, after the training period is
over, each user reports the maximum of its received power values
for the $S$ training beams and the corresponding beam index. Then,
the BS transmits a data stream to the user that has maximum
received power with the corresponding beam $\wbf_b$. In this case,
the rate $\Rc_S$ is given by
\begin{equation}
\Rc_S = \mathbb{E}\left[\log\left( 1 + \max_{1\le k \le K}\max_{1 \le b \le S} |\alpha_k|^2 M|\abf(\theta_k)^H\abf(\vartheta_b)|^2\right)\right].
\end{equation}
 First, consider
the case of $|\alpha_k|=1$ for all $k=1,\cdots,K$ as before. In
this case, we have the following theorem:

\vspace{0.5em}
\begin{theorem}\label{thm:thm_multibeam_sus}
For  $K=M^q$, $S=M^{\ell}$ and any $\ell,q \in (0,1)$ such that
$\ell+q <1$, we have asymptotic lower and upper bounds on $\Rc_S$
in the case of  $|\alpha_k|=1,~\forall~k$, given by
\begin{equation}   \label{eq:thm_multibeam_sus_main}
\log(1+M^{2q+2\ell - 1 - \epsilon}) ~\lesssim_M~
\mathbb{E}\left[\log(1+Z')\right] ~\lesssim_M~ \log(1+M^{2q+2\ell
- 1 + \epsilon})
\end{equation}
for any sufficiently small $\epsilon >0$, where $Z' = \max_k Z_k'$
and $Z_k'=\max_b M|\abf(\theta_k)^H\abf(\vartheta_b)|^2$.
\end{theorem}

\vspace{0.5em}

\begin{proof}
Proof consists of two steps as in the proofs of Lemma
\ref{lem:lemma1_bounds_of_Zk} and Theorem
\ref{thm:thm1_asymp_bounds}: (i) first, we bound
$\mathrm{Pr}\{Z_k' \le M^p\}$ and (ii) then bound
$\mathbb{E}\left[\log(1+Z')\right]$ using the bounds on
$\mathrm{Pr}\{Z_k' \le M^p\}$.

(i) First, consider $\mathrm{Pr}\{A\}:=\mathrm{Pr}\{Z_k'=\max_b M|\abf(\theta_k)^H\abf(\vartheta_b)|^2>M^p\}$
for $p \in (-1,1)$. Let $C_i$ be the event that $i=\arg\max_b
M|\abf(\theta_k)^H\abf(\vartheta_b)|^2$, i.e.,  $C_i$ is the event
that the $i$-th beam is the optimal beam for user $k$. Note that
the distribution of  $M|\abf(\theta_k)^H\abf(\vartheta_b)|^2 = M
\cdot F_M^2\left(\vartheta_b  - \theta_k\right) = M \cdot
F_M^2\left(\vartheta + \frac{2(b-1)}{S} - \theta_k\right)= M \cdot
F_M^2\left({\tilde{\theta}}_k+\frac{2(b-1)}{S}\right)$ is
independent of $b$ since  ${\tilde{\theta}}_k = \vartheta -
\theta_k \sim \mathrm{Unif}[-1,1]$ and the Fej\'er kernel
$F_M({\tilde{\theta}})$ is a periodic function with period $2$.
Hence, the events $C_1, C_2, \cdots, C_{M^\ell}$ are equally
probable as $\mathrm{Pr}\{C_i\}=\frac{1}{S}=\frac{1}{M^{\ell}}$
for every $i=1,\cdots,M^\ell$. Furthermore,  the conditional
events $A|C_1, A|C_2, \cdots, A|C_{M^\ell}$ are also equally
probable, i.e.,
$\mathrm{Pr}\{A|C_1\}=\cdots=\mathrm{Pr}\{A|C_{M^{\ell}}\}$ since
the situation is the same for each  $\vartheta_b$ due to the
periodicity of $F_M(\cdot)$ of period two and ${\tilde{\theta}}_k
\sim \mathrm{Unif}[-1,1]$. Hence,  by the law of total probability
and Bayes' rule, we have
\[
\mathrm{Pr}\{A\} = \sum_{i=1}^{M^{\ell}}
\mathrm{Pr}\{A|C_i\}\mathrm{Pr}\{C_i\}   = \mathrm{Pr}\{A|C_1\} =
M^{\ell}\cdot\mathrm{Pr}\{A,C_1\}.
\]
Thus, to bound  $\mathrm{Pr}\{A\}$, we need to bound
$\mathrm{Pr}\{A,C_1\}$. In order to bound
$\mathrm{Pr}\{A,C_1\}$, we find a sufficient condition for the
event $C_1$. Let $\tilde{C}_1(p)$ be the event
$M|\abf(\theta_k)^H\abf(\vartheta_1)|^2>M^p$. Then,  the event $\tilde{C}_1(p)$
with $p>2\ell - 1$ implies
\begin{equation}\label{eq:thm4_sus_proof}
|\theta_k -
\vartheta_1|\stackrel{(a)}{<}\frac{1}{\frac{\pi}{4}M^{(p+1)/2}}\stackrel{(b)}{=}\frac{1}{\frac{\pi}{4}M^{\ell+\delta/2}}
\end{equation}
for sufficiently large $M$, where $\delta = p-(2\ell - 1)>0$.
(Here, step (a) is  by \eqref{eq:lemma1_bounds_of_tildetheta}
of Lemma \ref{lem:lemma1_bounds_of_Zk} with $p \in (-1,1)$, and
step (b) is by the new additional condition $p>2\ell - 1$.)
Therefore, in this case, $|\theta_k -\vartheta_b|
>
|\frac{2}{M^\ell}-\frac{1}{\frac{\pi}{4}M^{\ell+\delta/2}}|>\frac{1}{M^{\ell}}=\frac{1}{2}\frac{2}{S}$
for any $b\neq 1$ and sufficiently large $M$, and this implies for sufficiently large $M$
\begin{equation} \label{eq:theo3proof51}
\tilde{C}_1(p) ~ \subset ~C_1 ~~\text{for}~~ p \in (-1,1) ~\mbox{and}~ p > 2\ell - 1.
\end{equation}
Now consider
 $\mathrm{Pr}\{A,C_1\}$
\begin{align}
\mathrm{Pr}\{A,C_1\} &=
\mathrm{Pr}\left\{\max_b M|\abf(\theta_k)^H\abf(\vartheta_b)|^2 > M^p, ~
1=\mathop{\arg\max}_b
M|\abf(\theta_k)^H\abf(\vartheta_b)|^2\right\} \nonumber\\
&=
\mathrm{Pr}\{\tilde{C}_1(p),C_1\}.  \label{eq:Theo3mid1}
\end{align}
By using \eqref{eq:theo3proof51} and  \eqref{eq:Theo3mid1}, we have
\begin{equation}
\mathrm{Pr}\{A,C_1\} \stackrel{(c)}{=} \mathrm{Pr}\{\tilde{C}_1(p),C_1\}
\stackrel{(d)}{=} \mathrm{Pr}\{\tilde{C}_1(p)\} =
\mathrm{Pr}\{M|\abf(\theta_k)^H\abf(\vartheta_1)|^2>M^p\}
\label{eq:Theorem3PAC}
\end{equation}
when  $p \in (-1,1)$ and $p>2\ell-1$ (these two conditions are required to apply \eqref{eq:theo3proof51}
for step (d), and step (c) is valid by \eqref{eq:Theo3mid1}). Now by applying \eqref{eq:lemma1_bounds_of_Zk} of
Lemma \ref{lem:lemma1_bounds_of_Zk} to the last term in
\eqref{eq:Theorem3PAC} and using
$\mathrm{Pr}\{A\}=M^{\ell}\cdot\mathrm{Pr}\{A,C_1\}$, we have for
 $p\in (-1,1)$ and $p>2\ell - 1$,
\begin{equation}
\frac{1}{2\pi M^{(1+p-2\ell)/2}} <\mathrm{Pr}\{A\}<
\frac{1}{\frac{\pi}{4}M^{(1+p-2\ell)/2}}.
\end{equation}

(ii) Substituting $\mathrm{Pr}\{Z_k' \le
M^p\}=1-\mathrm{Pr}\{A\}$ into $\mathrm{Pr}\{Z_k \le M^p\}$ in \eqref{eq:thm1_cdf_zk}
of the proof in Theorem \ref{thm:thm1_asymp_bounds} and following
the proof of Theorem \ref{thm:thm1_asymp_bounds}, we have for
$p=2q+2\ell-1-\epsilon$ with arbitrarily small $\epsilon > 0$,
\begin{equation}
\mathbb{E}[\log(1+Z')] \gtrsim_M \log(1+M^p),
\end{equation}
and for $p=2q+2\ell-1+\epsilon$ with arbitrarily small $\epsilon > 0$,
\begin{equation}
\mathbb{E}[\log(1+Z')] \lesssim_M \log(1+M^p)
\end{equation}
provided that $\ell+q<1$ (this is required for the condition
$p\in(-1,1)$), where $x\gtrsim_M y$ indicates
$\lim_{M\to\infty}x/y \ge 1$. Therefore, we have
\eqref{eq:thm_multibeam_sus_main}.  $\hfill \square$
\end{proof}

\vspace{0.5em}

\begin{corollary} \label{corol:multiRDBfraction}
For  $K=M^q$, $S=M^{\ell}$ and any $\ell,q \in (0,1)$ such that
$\frac{1}{2} < \ell+q <1$, we have
\begin{equation} \label{eq:multiRDBfraction}
\lim_{M\to\infty}\frac{\mathbb{E}[\log(1+Z')]}{\log(1+M)} = 2(q +
\ell) -1.
\end{equation}
\end{corollary}
\vspace{0.5em}

When the number $S$ of training beams is fixed, i.e., $\ell=0$,
Corollary \ref{corol:multiRDBfraction} reduces to the single beam result in  \eqref{eq:ratio_of_Rto_PerfCSI}.
  In the single user selection with
multiple training beams, as seen in \eqref{eq:multiRDBfraction},
the supremum of one for the achievable fraction can be achieved
arbitrarily closely by the combination of multiple users $q$ and
multiple training beams $\ell$.  Thus, when there exist not
sufficiently many users in the cell, multiple training beams can
be used to enhance the RDB performance. Note that even for $q=0$,
the optimal rate can be achieved with $\Rc_S$ by making $\ell
\uparrow 1$, as expected. (In fact, this case corresponds to the
case considered in the previous works on channel estimation for
sparse mm-wave MIMO channels, e.g.
\cite{Seo&Sung&Lee&Kim:15SPAWC}.)
 Note also that the effect of two terms
is not distinguishable at least in terms of the rate during the
data transmission period, although multiple training beams require
more training time. It can be shown that even with consideration of the random channel gain $\alpha_k \stackrel{i.i.d.}{\sim} \Cc\Nc(0,1)$, the same result as  (\ref{eq:multiRDBfraction}) is valid.

\subsection{The Multiple User Selection Case: Multiplexing Gain}

 In this section, we  consider multiple user
selection with RDB with multiple beams, aiming at multiplexing
gain, and investigate what can be achieved under the assumption of
$|\alpha_k|=1, ~\forall~k$ for simplicity. To do so, we
consider a simple user scheduling method based on the RBF method
\cite{Sharif&Hassibi:05IT}
 and then analyze the asymptotic performance of the
considered scheduling method, which gives an achievable
performance in the multi-beam multiple user selection case. The
considered scheduling method is basically the RBF scheme in
\cite{Sharif&Hassibi:05IT} with the $S$ random (asymptotically)
orthogonal beams given by \eqref{eq:equispaced_beam_def}. That is,
we choose a user that has  maximum
signal-to-interference-plus-noise ratio (SINR) for each beam
$\wbf_b$, $b=1,\cdots,S$, defined in
\eqref{eq:equispaced_beam_def}, and transmit $S$ independent data
streams to the $S$ selected users. In this case,  the received
signal of a selected user $\kappa_b$ is given by
\begin{equation}
y_{\kappa_b} = \sqrt{\frac{P_t}{S}}\hbf_{\kappa_b}^H\wbf_b +
\sqrt{\frac{P_t}{S}}\sum_{b^\prime \neq
b}\hbf_{\kappa_b}^H\wbf_{b^\prime} + n_{\kappa_b}, ~~b=1,\cdots,S,
\end{equation}
where $\kappa_b = \arg \max_{1 \le k \le K} \mathrm{SINR}_{k,b}$,
$ \mathrm{SINR}_{k,b} = \frac{\rho
M|\abf(\theta_k)^H\abf(\vartheta_b)|^2} {1+ \sum_{b'\neq b}\rho M
|\abf(\theta_k)^H\abf(\vartheta_{b'})|^2}$, and
$\rho=\frac{P_t}{S}$ is the per-user power of each scheduled user.
 The expected sum
rate of this scheduling method is given by
\begin{equation}\label{eq:sumrate_of_RM}
\Rc_M = \sum_{b=1}^S \Rc_{\kappa_b},
\end{equation}
where the data rate of each scheduled user $\kappa_b$ for beam $b$
is given by
\begin{equation}\label{eq:indivirate_of_RM}
\Rc_{\kappa_b}=  \mathbb{E}\left[ \log\left(1+ \max_{1\le k \le
K}\mathrm{SINR}_{k,b}\right)\right]=\mathbb{E}\left[\log\left(1+
\frac{\rho M|\abf(\theta_{\kappa_b})^H\abf(\vartheta_b)|^2} {1+
\sum_{b'\neq b}\rho M
|\abf(\theta_{\kappa_b})^H\abf(\vartheta_{b'})|^2}\right)\right].
\end{equation}
We first introduce the following lemma necessary to derive the
asymptotic result regarding \eqref{eq:sumrate_of_RM} and
\eqref{eq:indivirate_of_RM}:

\begin{lemma}\label{lem:lemma2_uppbdd_B}
For $|{\tilde{\theta}}_k| \in (0,1]$, we have an upper bound for
$F_M({\tilde{\theta}})$, given by $ F_M({\tilde{\theta}}_k) \le
\frac{1}{M|{\tilde{\theta}}_k|}$,  where $F_M(\cdot)$ is defined
in \eqref{eq:beam_patterngain}.
\end{lemma}

\begin{proof}
Since $F_M({\tilde{\theta}}_k)$ and $\frac{1}{M|{\tilde{\theta}}_k|}$ are even functions, it is enough
to consider ${\tilde{\theta}}_k \in (0,1]$ only.
From \eqref{eq:beam_patterngain}, we have an upper bound of $F_M({\tilde{\theta}}_k)$:
\begin{align*}
F_M({\tilde{\theta}}_k)
&\overset{(a)}{\le} \frac{1}{M} \frac{1}{\sin \frac{\pi {\tilde{\theta}}_k}{2}} \overset{(b)}{\le} \frac{1}{M{\tilde{\theta}}_k}
\end{align*}
where $(a)$ follows from $|\sin\frac{\pi {\tilde{\theta}}_k M}{2}| \le 1$ and
$\sin\frac{\pi {\tilde{\theta}}_k}{2}>0$ for ${\tilde{\theta}}_k \in (0,1]$, and $(b)$ follows from
\begin{align*}
&\frac{1}{{\tilde{\theta}}_k} - \frac{1}{\sin\frac{\pi
{\tilde{\theta}}_k}{2}} \ge 0 ~~\Longleftrightarrow~~
f({\tilde{\theta}}_k):=\sin\frac{\pi {\tilde{\theta}}_k}{2} -
{\tilde{\theta}}_k \ge 0.
\end{align*}
The RHS is true because $f(0)=f(1)=0$ with $f''({\tilde{\theta}}_k)=-\frac{\pi^2}{4}\sin\frac{\pi{\tilde{\theta}}_k}{2}<0$
for ${\tilde{\theta}}_k \in (0,1]$.
$\hfill \square$
\end{proof}
Now the following theorem shows the asymptotic result on
  \eqref{eq:sumrate_of_RM} and \eqref{eq:indivirate_of_RM} when the
total power $P_t$ is fixed regardless of $S$.

\vspace{0.5em}
\begin{theorem}[The case of fixed total transmit power $P_t=1$]\label{thm:thm_multi_user_selection}
For $K=M^q$, $S=M^{\ell}$  with $q \in (0,1)$ and $\ell \in
(0,q-\frac{\epsilon}{2})$,  asymptotic upper and lower bounds on
the per-user rate $\Rc_{\kappa_b}$ of selected user $\kappa_b$
for fixed total transmit power
$P_t=1$ are  given by
\begin{equation}
\log(1+M^{2q-1-\ell-\epsilon}) \lesssim_M \Rc_{\kappa_b}
\lesssim_M \log(1+M^{2q-1-\ell+\epsilon})
\end{equation}
for any $\epsilon >0$.
\end{theorem}
\vspace{0.5em}

\begin{proof}
The flow of proof is to first find lower and upper bounds on
$\Rc_{\kappa_b}$, denoted by $L$ and $U$, respectively, and then
to show that  the bounds $L$ and $U$ are asymptotically bounded as
\begin{equation}
\log(1+M^{2q-1-\ell-\epsilon}) \lesssim_M L \le \Rc_{\kappa_b} \le U \lesssim_M
\log(1+M^{2q-1-\ell+\epsilon}).
\end{equation}
To find $L$ and $U$, we consider a virtual user selection method
based on maximizing signal power not SINR for each beam
$\abf(\vartheta_b)$, i.e.,
\[
\tilde{\kappa}_b = \mathop{\arg \max}_{1 \le k \le K} M
|\abf(\theta_k)^H\abf(\vartheta_b)| ~~\text{for}~~ b=1,\cdots,S.
\]
Since the user $\tilde{\kappa}_b$ is chosen based on maximizing signal power only,
we have $\mathrm{SINR}_{\tilde{\kappa}_b,b} \le \mathrm{SINR}_{\kappa_b,b}$.
Therefore, a lower bound on $\Rc_{\kappa_b}$ can be obtained as
\begin{equation} \label{eq:lower_rdb_multibeam}
\Rc_{\kappa_b} \ge \Rc_{\tilde{\kappa}_b} = \mathbb{E}\left[\log\left(
1+\frac{\rho Z_{bb}}{1+ \rho \sum_{b^\prime \neq b}Z_{bb^\prime}}
\right)\right] =: L
\end{equation}
where $Z_{bb^\prime} :=
M|\abf(\theta_{\tilde{\kappa}_b})^H\abf(\vartheta_{b^\prime}))|^2$
for $b^\prime=1, \cdots, S$. Furthermore, an upper bound on
$\Rc_{\kappa_b}$ can be obtained by simply ignoring the inter-beam
interference as
\begin{equation} \label{eq:theoremUfirst}
\Rc_{\kappa_b} \le \mathbb{E}[\log(1+\rho M|\abf(\theta_{\kappa_b})^H\abf(\vartheta_{b})|^2 )]
\le \mathbb{E}[\log(1+\rho Z_{bb})] =: U.
\end{equation}
By modifying Theorem  \ref{thm:thm1_asymp_bounds}  to include
$\rho=1/S=M^{-\ell}$ in front of $Z_{bb}$ and applying the
modified theorem to $U$ in \eqref{eq:theoremUfirst}, we obtain
\begin{equation}\label{eq:theo5upperbound}
U \lesssim_M \log(1+M^{2q-1-\ell + \epsilon}).
\end{equation}
Hence, the claim on the upper bound follows. 

Now consider the case of lower bound $L$.
From the fact that $\mathbb{E}[f(X)]=\int f(x) p(x) dx \ge \int f(x) p(x,A) dx
= p(A)\mathbb{E}[f(X|A)]$ for a non-negative function $f(X)$, $L$ in \eqref{eq:lower_rdb_multibeam} with $\rho=1/S=M^{-\ell}$ can be bounded as
\begin{align}
&\mathbb{E}\left[\log\left(
1+\frac{M^{-\ell}Z_{bb}}{1+M^{-\ell} \sum_{b^\prime \neq b}Z_{bb^\prime}}
\right)\right] \nonumber \\
&~~\ge\mathrm{Pr}\{Z_{bb} \ge M^p\}\cdot
\mathbb{E}\left[\log\left(
1+\frac{M^{-\ell} Z_{bb}}{1+M^{-\ell}  \sum_{b^\prime\neq b}Z_{bb^\prime}}
\right) \bigg| Z_{bb} \ge M^p\right] \label{eq:thm_proof_lower_bdd_Rm}.
\end{align}
Under the condition that $\{Z_{bb} \ge M^p\}$, we have
\[
|\theta_{\tilde{\kappa}_b} - \vartheta_b| \le \frac{1}{\frac{\pi}{4}M^{(1+p)/2}}
\]
by  \eqref{eq:lemma1_bounds_of_tildetheta} of Lemma
\ref{lem:lemma1_bounds_of_Zk}. Therefore,
$|\theta_{\tilde{\kappa}_b}-\vartheta_{b^\prime}| >
\left|\frac{2}{S} - \frac{1}{\frac{\pi}{4}M^{(1+p)/2}}\right|$,
$\forall b^\prime \neq b$. Furthermore, we can re-arrange the
indices of $\{\vartheta_{b^\prime}\}_{b^\prime \neq b}$ in the
order of closeness to $\vartheta_b$ with the new indices $\{j\}$.
Then, we have
\begin{equation}\label{eq:lower_bdd_neighborhood_direc}
|\theta_{\tilde{\kappa}_b}- \vartheta_{2j-1}|, |\theta_{\tilde{\kappa}_b}- \vartheta_{2j}|
> \left|
\frac{2j}{M^{\ell}} - \frac{1}{\frac{\pi}{4}M^{(1+p)/2}}
\right|,
\end{equation}
since $\frac{2}{S}=\frac{2}{M^{\ell}}$ is the angular spacing
between two adjacent beams. We now have a lower bound on $L$,
given by
\begin{align}
L
&\overset{(a)}{\ge} \mathrm{Pr}\{Z_{bb} \ge M^p\} \cdot
\mathbb{E}\left[\log\left(
1+\frac{M^{-\ell}Z_{bb}}{1+M^{-\ell} \sum_{j\neq b}\frac{1}{M|\theta_{\tilde{\kappa}_b}-\vartheta_j|^2}}
\right) \bigg| Z_{bb} \ge M^p\right]  \nonumber \\
&\overset{(b)}{\ge} \mathrm{Pr}\{Z_{bb} \ge M^p\} \cdot
\mathbb{E}\left[\log\left(
1+\frac{M^{-\ell}Z_{bb}}{1+M^{-\ell} \sum_{j=1}^{\frac{S}{2}}\frac{2}{M\left|\frac{2j}{S}-\frac{1}{\frac{\pi}{4}M^{(1+p)/2}}\right|^2}}
\right) \bigg| Z_{bb} \ge M^p\right]  \nonumber \\
&\overset{(c)}{\gtrsim}_M \mathrm{Pr}\{Z_{bb} \ge M^p\} \cdot
\mathbb{E}\left[\log\left(
1+\frac{M^{-\ell}Z_{bb}}{1+2M^{\ell-1} \sum_{j=1}^{\frac{S}{2}}\frac{1}{j^2}}
\right) \bigg| Z_{bb} \ge M^p\right] \nonumber \\
&\overset{(d)}{\ge} \mathrm{Pr}\{Z_{bb} \ge M^p\} \cdot
\log\left(1+\frac{M^{-\ell}M^p}{1+\frac{\pi^2 M^{\ell-1}}{3}}\right) \nonumber \\
&\sim_M \log(1+M^{p-\ell}), \label{eq:theo5lowerbound}
\end{align}
where $(a)$  holds by  \eqref{eq:thm_proof_lower_bdd_Rm} and Lemma
\ref{lem:lemma2_uppbdd_B} with
$Z_{bj}=MF_M^2(\theta_{\tilde{\kappa}_b}-\vartheta_j)$; $(b)$ holds by
\eqref{eq:lower_bdd_neighborhood_direc}; $(c)$ holds because
$\left|\frac{2j}{M^{\ell}} -
\frac{1}{\frac{\pi}{4}M^{(1+p)/2}}\right|^2 \gtrsim_M
\left|\frac{j}{M^{\ell}}\right|^2$ for large $M$ provided that
$\ell < (1+p)/2$; $(d)$ follows from $\sum_{j=1}^{\infty}
\frac{1}{j^2}=\frac{\pi^2}{6}$; and  the last step holds because
$\mathrm{Pr}\{Z_{bb} \ge M^p\}\to 1$ for $p<2q-1$ by
\eqref{eq:asymp_prop_Z_k} and $\frac{\pi^2 M^{\ell-1}}{3} \to 0$
for $\ell <1$.  Hence the claim on the lower bound follows.

Note that the conditions used to derive
\eqref{eq:theo5lowerbound} and \eqref{eq:theo5upperbound} are $p <
2q -1$, $\ell < (1+p)/2$, and $\ell < 1$, and $q$ and $\ell$ are
given.  Set $p=2q-1 - \epsilon$ for $\epsilon >0$. Then, $\ell < q
- \epsilon/2 < 1$ since $q \in (0,1)$. This concludes proof.
$\hfill \square$
\end{proof}

\begin{theorem}[The case of fixed per-user power $\rho=1$, i.e. $P_t = S$]\label{cor:corollary_ofthm5}
For $K=M^q$ and $S=M^{\ell}$  with $q \in (0,1)$ and $\ell \in (0,\min(q-\frac{\epsilon}{2},\frac{1}{2}))$,
 asymptotic upper and lower bounds on
the per-user rate $\Rc_{\kappa_b}$ of selected user $\kappa_b$
for $\rho=1$
are given by
\begin{equation}
\log(1+M^{2q-1-\epsilon}) \lesssim_M \Rc_{\kappa_b} \lesssim_M
\log(1+M^{2q-1+\epsilon})
\end{equation}
for any $\epsilon >0$.
\end{theorem}

\begin{proof}
Proof is similar to that of Theorem \ref{thm:thm_multi_user_selection} and omitted due to space limitation.
$\hfill \square$
\end{proof}

\vspace{0.5em}

In Theorems \ref{thm:thm_multi_user_selection} and
\ref{cor:corollary_ofthm5}, the condition $\ell < q- \epsilon/2 <
1$ guarantees that  the $S$ beams are  asymptotically orthonormal
by \eqref{eq:multibeamortho}  and there exist more users than the
number of beams in the cell by the difference in the fractional
orders $q$ and $\ell$. First, note that the per-user rate in
Theorem \ref{cor:corollary_ofthm5} is the same as that in Theorem
\ref{thm:thm1_asymp_bounds}, i.e., the same per-user rate as that
of the single beam case can be achieved in the multi-beam
multi-user selection case when per-user power is fixed and the
same. Now consider the sum rate in the multiple beam multi-user
selection case. The sum rate $\Rc_M$ corresponding to Theorem
\ref{thm:thm_multi_user_selection} behaves as $\Rc_M =
\Theta(M^\ell \log M^{2q-1-\ell})$ when $2q-1-\ell > 0$. Pick
$\ell = 1-\delta_1$ for some small $\delta_1
>0$ and pick $q=\ell+\epsilon/2+\delta_2$ for some small $\delta_2
> 0$ such that $\delta_1 -\delta_2 -\epsilon/2 > 0$ to have $q<1$
and $-\delta_1 + 2\delta_2 +\epsilon > 0$ to have $2q-1-\ell > 0$.
Then, we have $M^\ell \log M^{2q-1-\ell} = M^{1-\delta_1} \log
(M^{-\delta_1 + 2\delta_2 +\epsilon})$. Thus, sum rate behavior
arbitrarily close to linear scaling w.r.t. the number of antennas
is possible in the multi-beam multi-user selection case by random
beamforming (randomly-directional beamforming) with proper user
scheduling under the UR-LoS channel model. This is a significant
difference from the sum rate behavior \eqref{eq:RBF_sumrate} of
the RBF method \cite{Sharif&Hassibi:05IT} in large-scale MIMO,
i.e., $\lim_{K\rightarrow\infty} \frac{\log K}{M} =0$, under the
i.i.d. Rayleigh fading channel model \eqref{eq:richscatteringgkj}
representing rich scattering environments. The major performance
difference results from the difference in degrees-of-freedom in
the two channels: the UR-LoS channel \eqref{eq:ChModelURLOS} and
the i.i.d. Rayleigh fading channel \eqref{eq:richscatteringgkj}.
In the i.i.d. Rayleigh fading channel case, we have $M$
independent parameters and the channel vector is randomly located
within a ball in the $M$-dimensional space. Consider a cone around
each axis in the $M$-dimensional space so that channel vectors
each of which is contained in each of the cones are roughly
orthogonal, as shown in Fig. 3 of \cite{Lee&Sung:14ITsub}. Then,
the probability that a
 channel vector  generated randomly according to
\eqref{eq:richscatteringgkj} falls into such a cone is
exponentially decreasing as $M$ increases (See Appendix
\ref{appen:cone_probability}). Hence, if the number $K$ of users
randomly distributed within the ball does not increases
exponentially fast w.r.t. the dimension $M$ (i.e.,
$\lim_{K\rightarrow\infty} \frac{\log K}{M} =0$), it is difficult
to find $M$ users whose channel vectors are contained in the $M$
roughly-orthogonal cones (one for each)
\cite{Sharif&Hassibi:05IT,TomasoniEtAl:09ISIT,Hur&Tulino&Caire:12IT}
(the goal of  SUS \cite{Yoo&Goldsmith:06JSAC}, RBF
\cite{Sharif&Hassibi:05IT} or ReDOS-PBR \cite{Lee&Sung:14ITsub}
scheduling is to find such $M$ users\footnote{This is why it is
not easy to apply SUS, RBF, or ReDOS-PBR to finding roughly
orthogonal simultaneous users more than four to six in practical
 setup.}), and linear sum rate scaling by random beamforming w.r.t. the
dimension $M$ (i.e., the number of antennas) is not attainable. In
the considered mm-wave MIMO with the UR-LoS channel model,
however,
 the situation is quite different.  Theorems
\ref{thm:thm_multi_user_selection} and \ref{cor:corollary_ofthm5}
state that  sum rate scaling arbitrarily close to linear scaling
w.r.t. $M$ is possible in this case. This is because the
degree-of-freedom in the UR-LoS channel model
\eqref{eq:ChModelURLOS}  with $\alpha_k=1$ is {\em one} regardless
of the value of $M$. The orthogonality of the multiple transmit
beams is attained by simply dividing the line of the normalized
angle $\theta$ with length 2 by line segments each with length
$2/S=2/M^\ell$. Thus, if $K=M^q$ with $q > \ell$, there exists
many users in each line segment one of which is well matched to
the transmit beam direction associated with each line segment if
$2q-1-\ell > 0$. Thus, in this case user scheduling to select such
$S$ users is beneficial for random beamforming-based BS operation.
In fact, the channel matrix composed of the channel vectors of the
users scheduled in such a way satisfies the asymptotically
favorable propagation condition  in
\cite{Ngo&Larsson&Marzetta:14EUSIPCO}. Note that the fundamental
difference between the UR-LoS channel \eqref{eq:ChModelURLOS}
modeling high propagation directivity and the i.i.d. Rayleigh
fading channel model \eqref{eq:richscatteringgkj} for rich
scattering is that linear sum rate scaling w.r.t. the number $M$
of antennas by random beamforming is attainable with the number
$K$ of users increasing linearly w.r.t. $M$ in the UR-LoS channel
model, whereas linear sum rate scaling w.r.t. the number $M$ of
antennas by random beamforming is attainable with $K$ increasing
exponentially w.r.t. $M$ in the i.i.d. Rayleigh fading channel
model! Thus, high directivity is preferred to rich scattering for
opportunistic  random beamforming under massive MIMO situation.
This suggests that opportunistic random beamforming is a viable
choice for massive MIMO in the mm-wave band with high propagation
directivity.

\subsection{Performance comparison: The fractional rate order}

In this subsection, we compare the asymptotic performance of the
three schemes considered in the previous sections. Here, we assume
$\alpha_k=1, ~\forall~k$  and $P_t=1$. In order to compare the
relative performance, we define the {\em fractional rate order
(FRO) $\gamma$} as
\begin{equation}
\gamma:= \lim_{M\to\infty}\frac{\log \Rc}{\log M}.
\end{equation}
Note that $\Rc = \Theta(M^\gamma)$ for $\gamma \ne 0$. For $\gamma
>0$, $\Rc$ increases to infinity as $M \rightarrow \infty$,
whereas for $\gamma <0$, $\Rc$ decreases to zero as $M\rightarrow
\infty$.
 Now consider the three rates $\Rc=\Rc_1,\Rc_S$, and $\Rc_M$.
First,  for  the single beam RDB rate $\Rc_1$ we have
by Theorem \ref{thm:thm1_asymp_bounds}
that
\begin{equation}\label{eq:perf_region_R1}
\gamma_1=\lim_{M\to\infty}\frac{\log \Rc_1}{\log M}
=\left\{
\begin{array}{ll}
0, & ~\text{for}~ q \in (\frac{1}{2},1) \\
2q-1 & ~\text{for}~ q \in (0,\frac{1}{2}),
\end{array}
\right.
\end{equation}
where we used $\log(1+x)=x$ for small $x$ for the second part.
Next, for the multi-beam RDB scheme with  single-user selection,
we have
\begin{equation}\label{eq:perf_region_RS}
\gamma_S=\lim_{M\to\infty}\frac{\log \Rc_S}{\log M}
=0, ~~ \text{for} ~~ q \in(0,1)
\end{equation}
by Theorem \ref{thm:thm_multibeam_sus} with setting $\ell$ such
that $1/2 < \ell + q < 1$. Here, $\gamma_S=0$ is achieved even for
$q\in (0,1/2)$ because of added $\ell$. Finally, we consider the
multi-beam RDB strategy with multi-user selection. In this case,
$\Rc_M = \Theta( M^{\ell} \log(1+M^{2q-\ell-1}))$ from Theorem
\ref{thm:thm_multi_user_selection} and \eqref{eq:sumrate_of_RM}.
Using $M^\ell \log(1+M^{2q-\ell-1})=
\log(1+1/M^{-2q+\ell+1})^{M^\ell M^{-2q+\ell+1-(-2q+\ell+1)}} =
\log(1+\frac{1}{M^{-2q+\ell+1}})^{M^{-2q+\ell+1} M^{2q-1}} \sim_M
M^{2q-1}$ by setting $\ell$ such that $2q-1< \ell < q$, we obtain
\begin{equation}\label{eq:perf_region_RM}
\lim_{M\to\infty}\frac{\log \Rc_M}{\log M}
=2q-1, ~~ \text{for} ~~ q \in(0,1).
\end{equation}
\begin{figure}[h]
\begin{psfrags}
        \psfrag{1}[c]{\large  $1$} %
        \psfrag{-1}[c]{\large  $-1$} %
        \psfrag{q}[c]{ \large $q$} %
        \psfrag{p}[r]{
         \large $\underset{M \to \infty}{\lim} \frac{\log\Rc}{\log M}$ \hspace{-1em}} %
        \psfrag{h}[c]{\large  $\frac{1}{2}$} %
        \psfrag{0}[c]{\large  $0$} %
        \psfrag{r1}[l]{\large  $\Rc=\Rc_1$} %
        \psfrag{rs}[l]{\large  $\Rc=\Rc_S$} %
        \psfrag{rm}[l]{\large  $\Rc=\Rc_M$} %
        \centerline{ \scalefig{0.45} \epsfbox{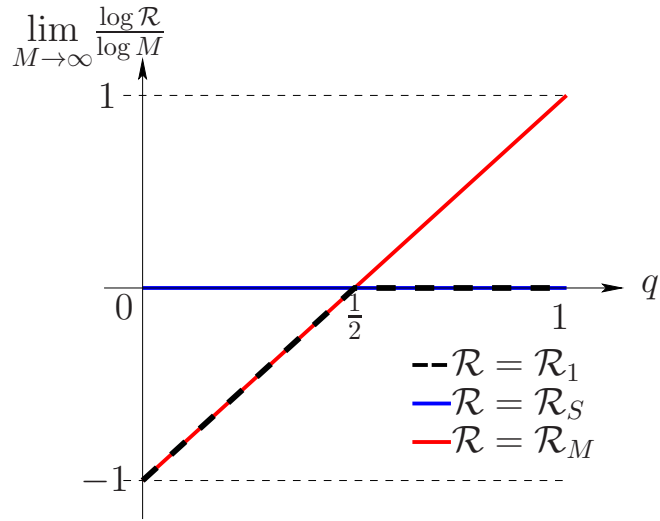} }
        \caption{Fractional rate order versus $q$}
        \label{fig:perf_region}
\end{psfrags}
\end{figure}
Fig. \ref{fig:perf_region} shows \eqref{eq:perf_region_R1},
\eqref{eq:perf_region_RS} and \eqref{eq:perf_region_RM} versus $q
\in (0,1)$, and shows which strategy among RDB should be used for
different $q$ determining the number of users in the cell relative
to the number of antenna elements. $\Rc_M$ has the largest FRO for
$q\in(\frac{1}{2},1)$, whereas $\Rc_S$ has the largest FRO for
$q\in(0,\frac{1}{2})$.  $\gamma_1$ is a lower bound on both
$\gamma_S$ and $\gamma_M$ for all $q \in (0,1)$, and $\gamma_M
\uparrow 1$ as $q \uparrow 1$, as mentioned already. Note that
$\gamma_M < 0$ for $q \in (0,\frac{1}{2})$, which implies $\Rc_M
\rightarrow 0$ as $M\rightarrow \infty$. This is because the total
number of users in the cell is not sufficient to find a user well
matched to each beam. The transition point of determining the
scarcity of users in the cell is $K=\Theta(\sqrt{M})$ under the
UR-LoS model, whereas the transition point is
$K=\Theta(\exp(\eta^2 M))$ for some $\eta^2 \in (0,1)$ or
equivalently $M=\Theta(\log K)$ in the i.i.d. Rayleigh fading
channel model. When $q \in (0,1/2)$, i.e., there exist not many
users in the cell, the best strategy is the multi-beam single-user
strategy. This implies that downlink channel estimation to
identify the user channel, i.e., user's propagation angle is
important in this regime. On the other hand, when $q \in (1/2,1)$,
i.e., there exist a sufficient number of users in the cell,
downlink channel estimation is less important, and user scheduling
based on equi-spaced random beams with an arbitrary angle offset
is sufficient to obtain good performance and achieves linear sum
rate scaling w.r.t. the number of antennas when the number of
users increases linearly w.r.t. the number of antennas.

\section{Numerical results}
\label{sec:NumericalResult}

In this section, we provide some numerical results
to validate our asymptotic analysis in the previous sections.
All the expectations in the below are average over 5000
channel realizations and we set to $P_t=1$.

\subsection{The Single Beam Case}

To verify the asymptotic analysis in Section
\ref{sec:singlerandombeam}, we considered a mm-wave MU-MISO
downlink system with the UR-LoS channel model. Fig.
\ref{fig:singlebeam_ratio} (a) and (b) shows the value of
$\frac{\mathbb{E}[1+Z]}{\log(1+M)}$ versus $q$ for
$M=100,500,1000,5000,10000$ for $\alpha_k=1$ and $\alpha_k \sim
\Cc\Nc(0,1)$, respectively. It is seen that the curve of
$\frac{\mathbb{E}[1+Z]}{\log(1+M)}$ versus $q$  gradually
converges to the theoretical line of $2q-1$ for $q>\frac{1}{2}$
and $0$ for $q \le \frac{1}{2}$ as $M$ increases. Note that there exist some gap between the theoretical asymptotic line and the finite-sample results. This results from the slow rate of convergence. 
\begin{figure}[ht]
\centerline{ \SetLabels
\L(0.25*-0.1) (a) \\
\L(0.75*-0.1) (b) \\
\endSetLabels
\leavevmode
\strut\AffixLabels{
\scalefig{0.5}\epsfbox{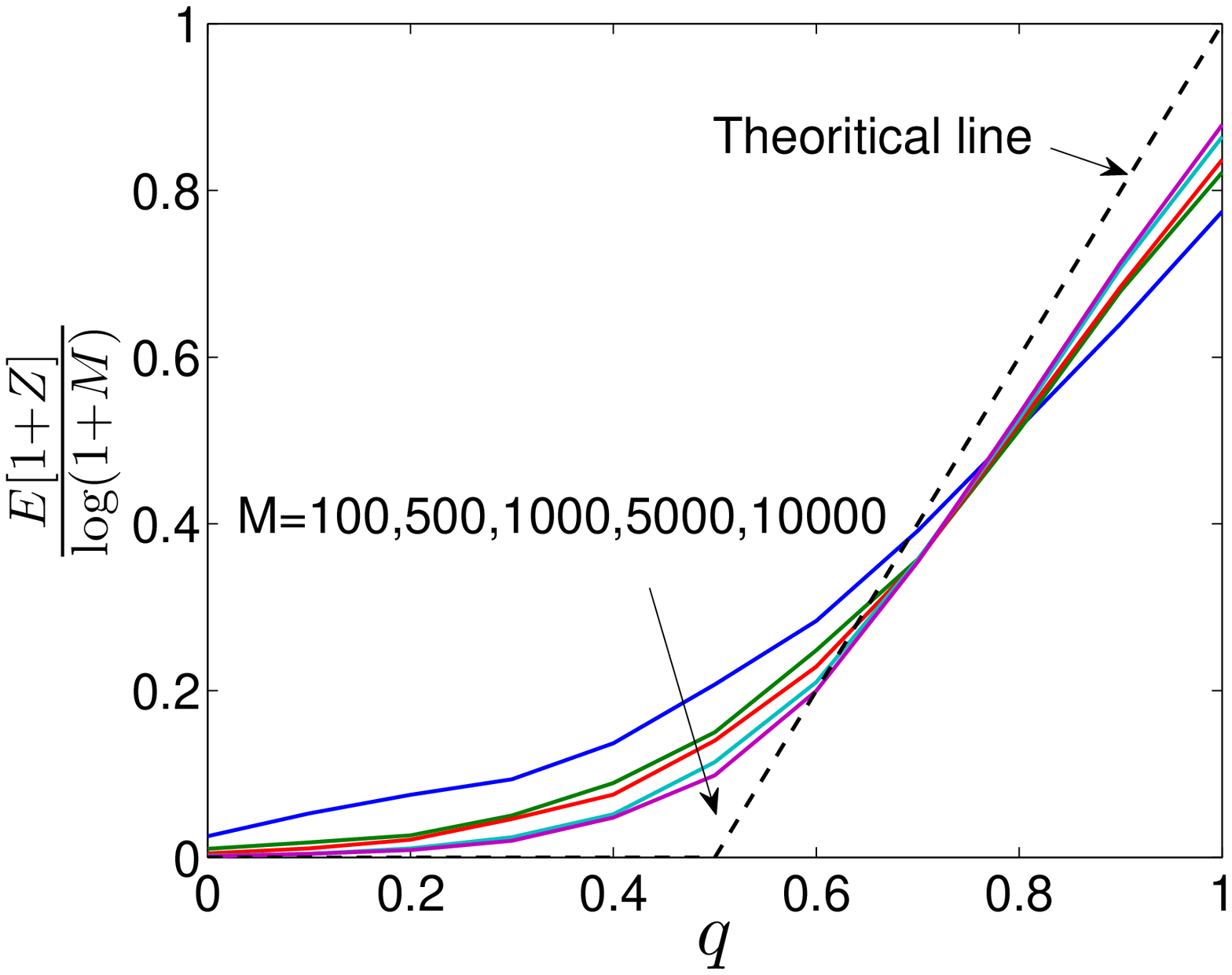}
\scalefig{0.5}\epsfbox{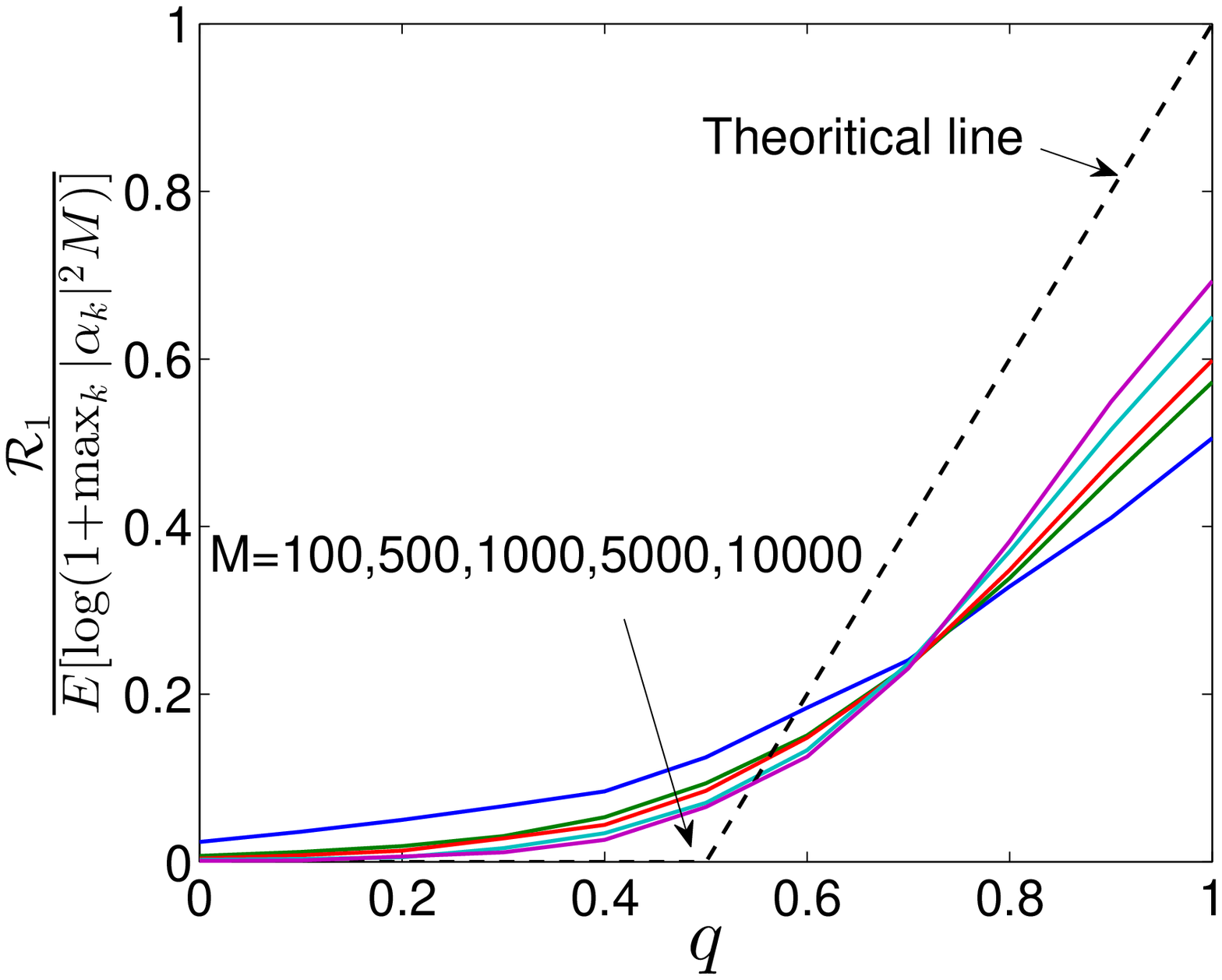}
} } \vspace{0.5cm} \caption{The ratio of the  RDB rate $\Rc_1$ to
the rate with perfect CSI $\mathbb{E}[\log(1+\max_k |\alpha_k|^2
M)]$ versus $q$ for different $M$: (a) $\alpha_k=1$ and (b)
$\alpha_k\sim\Cc\Nc(0,1)$} \label{fig:singlebeam_ratio}
\end{figure}
Fig. \ref{fig:singlebeam_rate} (a) and (b) show the actual RDB
rate  w.r.t. $M$  for $q=0.1$ to $0.5$ and $q=0.6$ to $1$,
respectively, in the case of $\alpha_k \sim \Cc\Nc(0,1)$. It is
seen in Fig. \ref{fig:singlebeam_rate} (a) that the RDB  rate for
$q$ below $0.5$ decreases as $M$ increases, but it almost remains
the same when $q=0.5$. On the other hand, it is seen in Fig.
\ref{fig:singlebeam_rate} (b) that  the RDB rate for  $q$ above
$0.5$ increases as $M$ increases. (Since x-axis is in log scale,
the rate curve is linear as expected by Theorem
\ref{thm3:ratio_of_Rto_PerfCSI} when $q>0.5$.) The results  in
Figs.  \ref{fig:singlebeam_ratio} and \ref{fig:singlebeam_rate}
coincide with Theorems \ref{thm:thm1_asymp_bounds} and
\ref{thm3:ratio_of_Rto_PerfCSI}.
\begin{figure}[ht]
\centerline{ \SetLabels
\L(0.25*-0.1) (a) \\
\L(0.75*-0.1) (b) \\
\endSetLabels
\leavevmode
\strut\AffixLabels{ \scalefig{0.5}\epsfbox{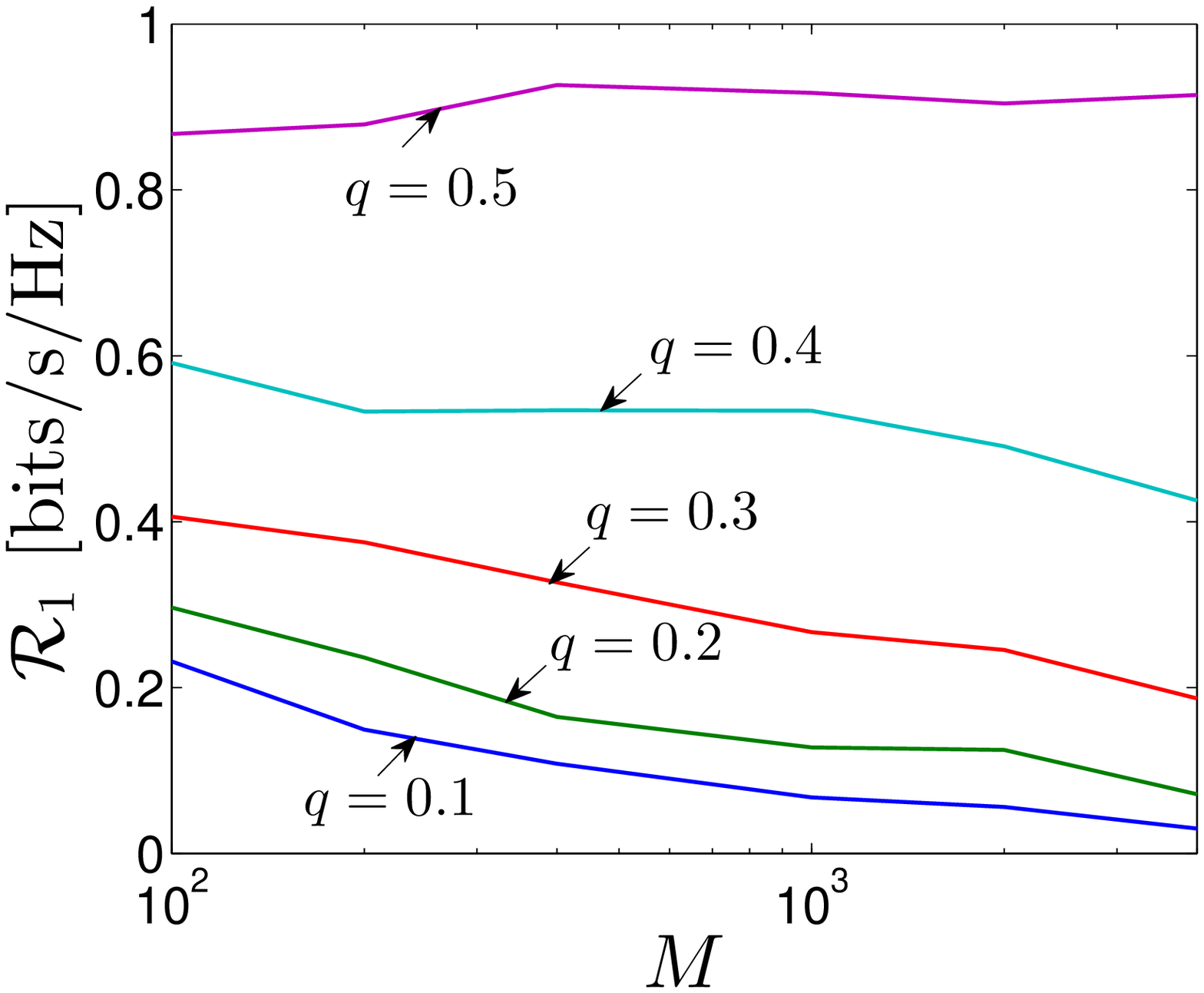}
\scalefig{0.5}\epsfbox{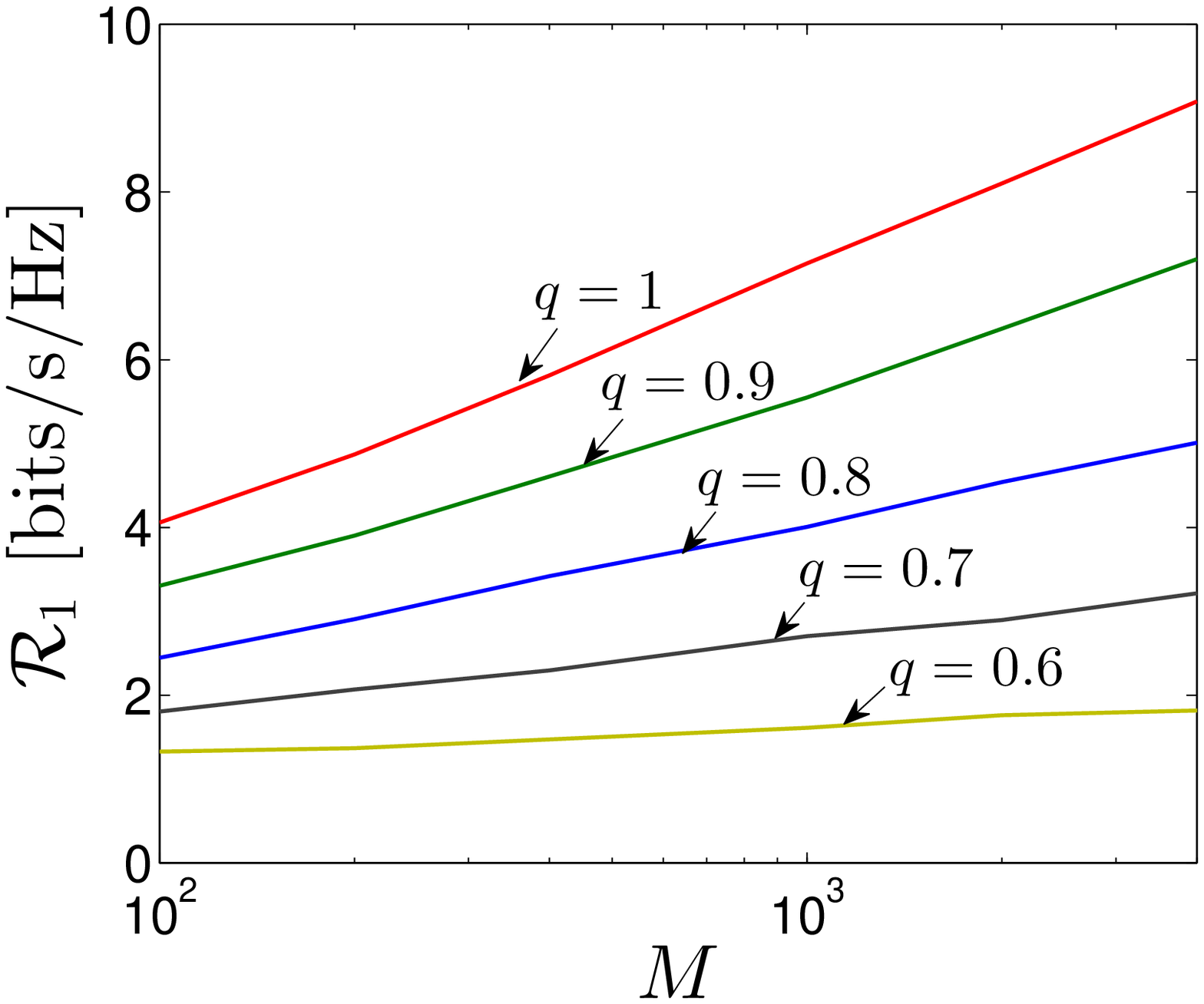} } }
\vspace{0.5cm} \caption{The RDB rate $\Rc_1$  versus $M$ with $\alpha_k\sim\Cc\Nc(0,1)$ for different $q$ (log scale on x-axis): (a) $q=0.1,0.2,\cdots,0.5$ and
(b) $q=0.6,0.7,\cdots,1$ } \label{fig:singlebeam_rate}
\end{figure}

\subsection{The Multiple Beam Case}

We first considered the multiple beam RDB with single user selection.
Fig. \ref{fig:multibeam_sus_ratio} (a) and (b) show
the ratio of the multiple beam RDB rate $\Rc_S$  with single-user selection to the rate with perfect CSI versus $q$ for different $\ell$ in the cases of $\alpha_k=1$ and $\alpha_k\sim\Cc\Nc(0,1)$, respectively, when $M=1000$.
It is seen that the simulation curves
  roughly match the theoretical lines.
 \begin{figure}[ht]
\centerline{ \SetLabels
\L(0.25*-0.1) (a) \\
\L(0.75*-0.1) (b) \\
\endSetLabels
\leavevmode
\strut\AffixLabels{ \scalefig{0.5}\epsfbox{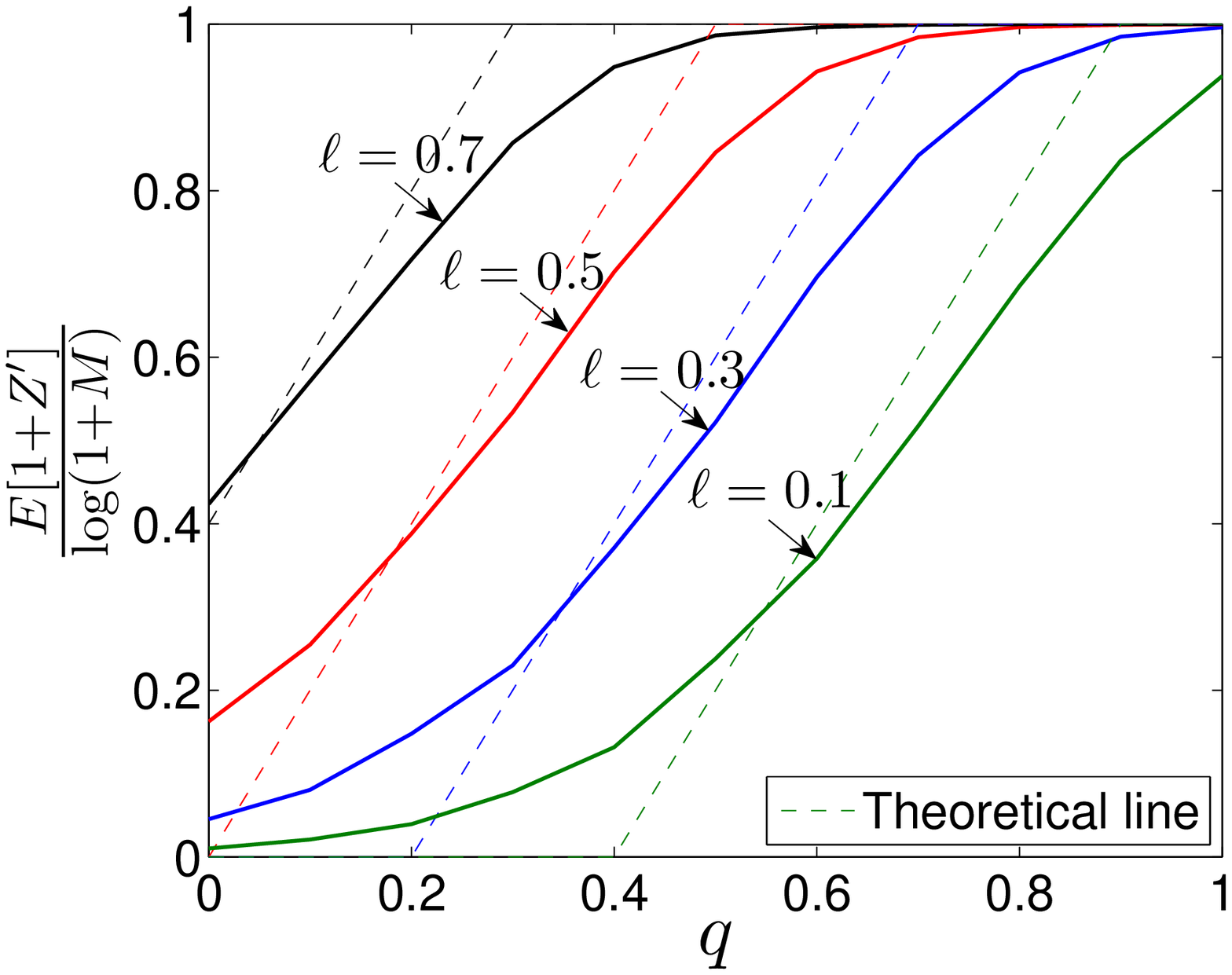}
\scalefig{0.5}\epsfbox{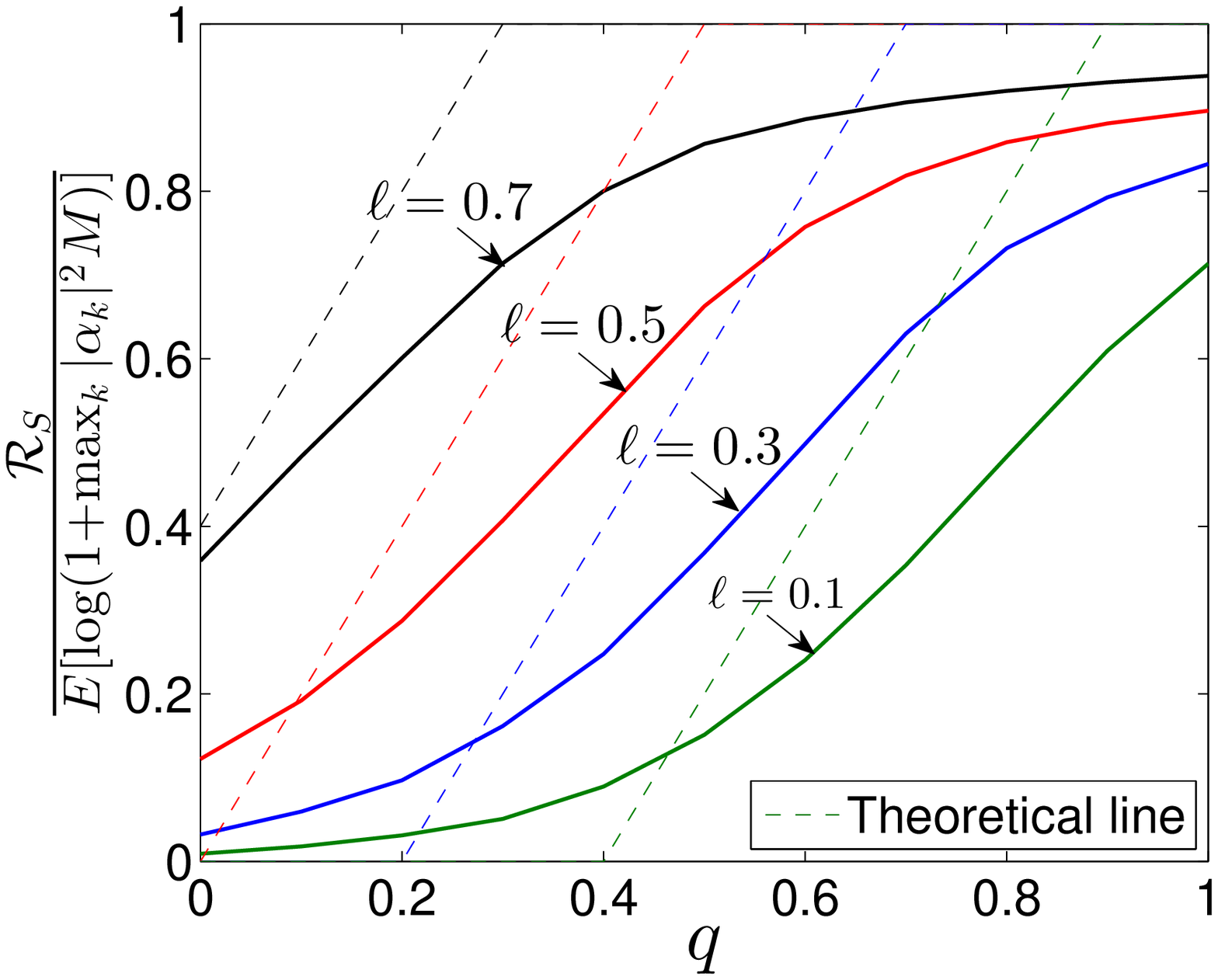} } }
\vspace{0.5cm} \caption{The ratio of $\Rc_S$ to the rate with perfect CSI
$\mathbb{E}[\log(1+\max_k|\alpha_k|^2M)]$ versus $q$ for different $\ell$: (a) $\alpha_k=1, \forall k$
and (b) $\alpha_k \sim \Cc\Nc(0,1), \forall k$} \label{fig:multibeam_sus_ratio}
\end{figure}
We then verified the rate $\Rc_S$ for $q = 0.3$ with different
$\ell$. It is seen in Fig. \ref{fig:multibeamCombine} (a) that
$\Rc_S$ increases as $M$ increases for the cases of $\ell>0.2$
(i.e., $q+\ell > 0.5$), as predicted by Theorem
\ref{thm:thm_multibeam_sus}. On the other hand, the rate decreases
for the case of $\ell < 0.2$ as $M$ increases. Finally, we
verified the multi-beam multi-user selection RDB. We set to
$q=0.7$ and used $\alpha_k \sim \Cc\Nc(0,1)$, $\forall~k$. Fig.
\ref{fig:multibeamCombine} (b) shows the per-user rate
$\Rc_{\kappa_b}$ in Theorem \ref{thm:thm_multi_user_selection}
versus $M$ for different $\ell$. It is seen that the per-user rate
$\Rc_{\kappa_b}$ increases when $\ell < 0.4$, whereas it decreases
when $\ell > 0.4$, as $M$ increases, as predicted by Theorem
\ref{thm:thm_multi_user_selection} (i.e., $2q-1-\ell > 0$ or
$2q-1-\ell < 0$).

 \begin{figure}[ht]
\centerline{ \SetLabels
\L(0.25*-0.1) (a) \\
\L(0.75*-0.1) (b) \\
\endSetLabels
\leavevmode
\strut\AffixLabels{
\scalefig{0.5}\epsfbox{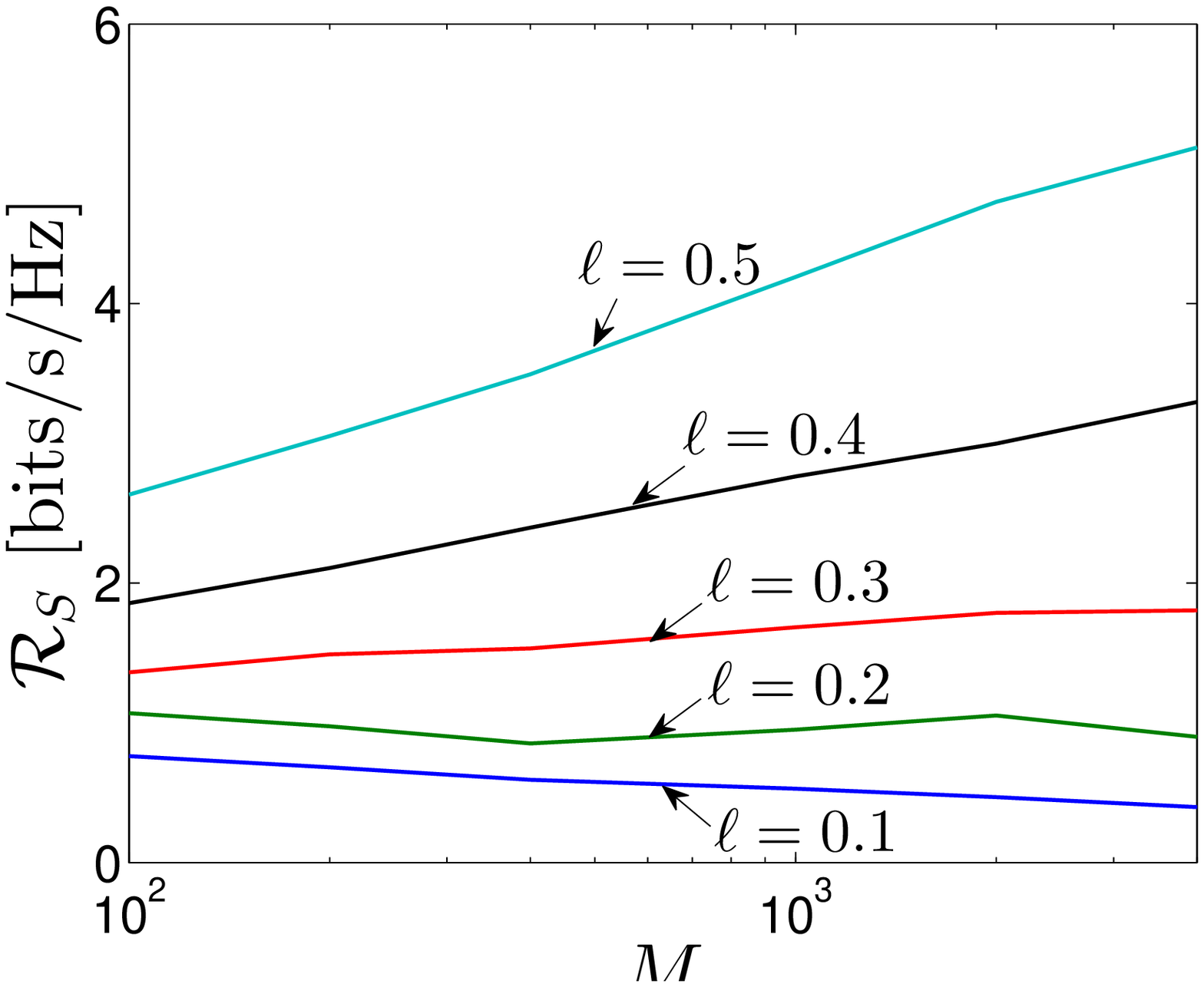}
\scalefig{0.5}\epsfbox{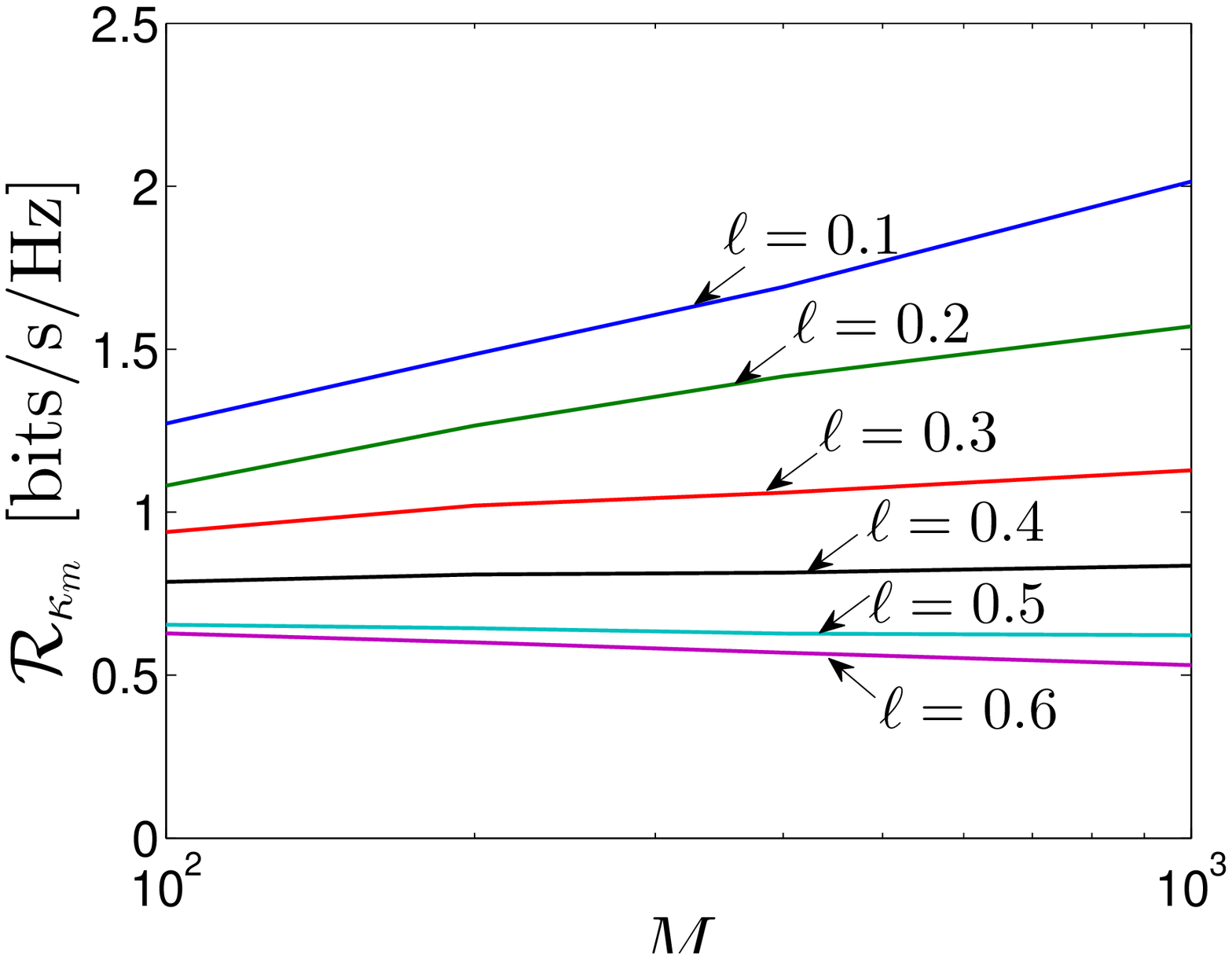}
} } \vspace{0.5cm} \caption{(a) $\Rc_S$  versus $M$ ($q=0.3$ and
$\alpha_k \sim \Cc\Nc(0,1)$) (log scale on x-axis) and (b) the
per-user rate of the multi-beam multi-user selection RDB versus
$M$ ($q=0.7$ and $\alpha_k \sim \Cc\Nc(0,1)$)}
\label{fig:multibeamCombine}
\end{figure}

\section{Conclusion}
\label{sec:conclusion}

We have considered RDB for millimeter-wave MU-MISO and examined
the associated MU gain, using asymptotic performance analysis
based on the UR-LoS channel model which well captures radio
propagation channels in the mm-wave band. We have shown that there
exists a transition point on the number of users relative to the
number of antenna elements for non-trivial performance of the RDB
scheme and have identified the case in which downlink training and
channel estimation are important for good performance. We have
also shown that sum rate scaling arbitrarily close to linear
scaling w.r.t. the number of antenna elements can be achieved
under the UR-LoS channel model  by random beamforming based on
multiple beams equi-spaced in the angle domain and proper user
scheduling, if the number of users in the cell increases linearly
w.r.t. the number of antenna elements. We have compared three RDB
schemes composed of beamforming and user scheduling based on the
newly defined fractional rate order, yielding insights into the
most effective beamforming and scheduling choices for mm-wave
MU-MISO in various operating conditions. Simulation results
validate the analysis based on asymptotic techniques for finite
cases. The results here is based on the simplified UR-LoS
channel model capturing high propagation directivity, and thus
extension to a general channel model is left as future work.

\appendices

\section{Distribution of $\vartheta - \theta_k$}  \label{append:uniformdist}

Since $\vartheta,\theta_k \overset{\text{i.i.d.}}{\sim} \mathrm{Unif}[-1,1]$,
the difference random variable ${\tilde{\theta}}_k$ has the distribution, given by
\begin{equation}
p({\tilde{\theta}}) = \left\{
\begin{array}{ll}
\frac{1}{4}{\tilde{\theta}} + \frac{1}{2}, & -2 \le {\tilde{\theta}} \le 0 \\
-\frac{1}{4}{\tilde{\theta}} + \frac{1}{2}, &~~ 0 \le {\tilde{\theta}} \le 2. \\
\end{array}
\right.
\end{equation}
For any function $f(\tilde{\theta})$ with the periodicity of period two,
we have $f(\tilde{\theta}) = f(\tilde{\theta} + 2)$ for
$\tilde{\theta} \in [-1,0]$ and $f(\tilde{\theta}) = f(\tilde{\theta} - 2)$ for
$\tilde{\theta} \in [0,1]$.
Therefore,
we can regard $p(\tilde{\theta})$ on the function $f(\tilde{\theta})$
as
\begin{equation}
p(\tilde{\theta}) =
\left\{
\begin{array}{ll}
\frac{1}{4}\tilde{\theta}+\frac{1}{2} - \frac{1}{4}\tilde{\theta}, & ~ -1 \le \tilde{\theta} \le 0 \\
-\frac{1}{4}\tilde{\theta} + \frac{1}{2} + \frac{1}{4}\tilde{\theta}, & ~ ~~~ 0\le
\tilde{\theta} \le 1
\end{array}
\right.
\end{equation}
i.e., $\tilde{\theta} \sim \mathrm{Unif}[-1,1]$.

\section{Proof of Theorem \ref{thm3:ratio_of_Rto_PerfCSI}}
\label{appen:thm3}

Before proving Theorem \ref{thm3:ratio_of_Rto_PerfCSI}, we prove another interesting lemma of which proof is partly used in proof of Theorem \ref{thm3:ratio_of_Rto_PerfCSI}.

\begin{lemma}
\label{thm2:ratio_of_R}
For $K = M^q$, $q \in (\frac{1}{2},1)$ and $\alpha_k \overset{\text{i.i.d.}}{\sim} \Cc\Nc(0,1)$,
we have
\begin{equation}
\lim_{M \to \infty}\frac{\Rc_1}{\mathbb{E}\left[\log(1+|\alpha_{k'}|^2Z_{k'})\right]} = 1,
\end{equation}
where  $k' = \arg \max_k Z_k$, and $\Rc_1$ is the optimal RDB rate in \eqref{eq:expected_rate_s1} considering the random path gain.
\end{lemma}

\begin{proof}
$\Rc_1$ is bounded as
\begin{equation}\label{eq:appen1_bounds}
\mathbb{E}[\log (1+|\alpha_{k'}|^2 Z_{k'})]
\le \Rc_1 \le \mathbb{E}\left[
\log\left(1+\left(\max_k |\alpha_k|^2\right) Z_{k'}\right)
\right].
\end{equation}
Eq. \eqref{eq:thm1_asymp_bounds} in Theorem \ref{thm:thm1_asymp_bounds} can easily be modified to
\begin{equation}   \label{eq:Append2Lemma3ULB1}
\log(1+\beta M^{2q-1-\epsilon}) ~\lesssim_M~ \mathbb{E}\left[\log(1+\beta Z)\right] ~\lesssim_M~
\log(1+\beta M^{2q-1+\epsilon})
\end{equation}
for $q \in (\frac{1}{2},1)$ and $\beta > 0$.  Note that
$\mathbb{E}[\log (1+|\alpha_{k'}|^2 Z_{k'})]= \mathbb{E}[\mathbb{E}[\log (1+|\alpha_{k'}|^2 Z_{k'})~|~|\alpha_{k'}|^2 ]]$ by the law of iterated expectations.  Applying the lower bound in \eqref{eq:Append2Lemma3ULB1} to $\mathbb{E}[\log (1+|\alpha_{k'}|^2 Z_{k'})~|~|\alpha_{k'}|^2 ]$, we have
\begin{equation}
\mathbb{E}[\log(1+|\alpha_{k'}|^2 M^{2q-1-\epsilon})] \lesssim_M
\mathbb{E}[\log (1+|\alpha_{k'}|^2 Z_{k'})].
\end{equation}
For $q \in (\frac{1}{2},1)$, we have
\begin{align}
\mathbb{E}[\log(1+|\alpha_{k'}|^2 M^{2q-1-\epsilon})]
&\sim_M \mathbb{E}[\log(|\alpha_{k'}|^2 M^{2q-1-\epsilon})] \nonumber \\
&=\mathbb{E}[\log|\alpha_{k'}|^2] + (2q-1-\epsilon)\log M \nonumber\\
& \sim_M  (2q-1-\epsilon)\log M  \label{eq:appen1_lower}
\end{align}
for any sufficiently small $\epsilon>0$ such that $2q-1-\epsilon >0$. Since $|\alpha_{k'}|^2 \sim \chi^2(2)$, $\mathbb{E}[\log|\alpha_{k'}|^2]$ is a constant.

Now consider the upper bound in \eqref{eq:appen1_bounds}. Again applying the law of iterated expectations and the upper bound in \eqref{eq:Append2Lemma3ULB1}, we have $\mathbb{E}\left[
\log\left(1+\left(\max_k |\alpha_k|^2\right) Z_{k'}\right)
\right] \le \mathbb{E}\left[
\log\left(1+\left(\max_k |\alpha_k|^2\right) M^{2q-1+\epsilon}\right)
\right]$.  From the fact that
$\mathbb{E}[\log(1+\max_k|\alpha_k|^2)] \sim_M \log(\log K)$ \cite{Sharif&Hassibi:05IT}, the above bound can further be simplified as
\begin{align}
\mathbb{E}\left[
\log\left(1+\left(\max_k |\alpha_k|^2\right) Z_{k'}\right)
\right] &\lesssim_M \log(M^{2q-1+\epsilon}\log K) \nonumber\\
&\sim_M (2q-1+\epsilon)\log M + \log(\log M). \label{eq:appen1_upper}
\end{align}
Dividing \eqref{eq:appen1_bounds} by
$\mathbb{E}[\log (1+|\alpha_{k'}|^2 Z_{k'})]$, we have
\begin{align}
1 \le \frac{\Rc_1}{\mathbb{E}[\log (1+|\alpha_{k'}|^2 Z_{k'})]} &\le \frac{\mathbb{E}\left[
\log\left(1+\left(\max_k |\alpha_k|^2\right) Z_{k'}\right)
\right]}{\mathbb{E}[\log (1+|\alpha_{k'}|^2 Z_{k'})]} \nonumber\\
&\overset{(a)}{\lesssim}_M \frac{(2q-1+\epsilon)\log M + \log\log M}{ (2q-1-\epsilon)\log M} \nonumber \\
&\sim_M \frac{2q-1+\epsilon}{2q-1-\epsilon} \label{eq:appen1_final}
\end{align}
where step $(a)$ follows from \eqref{eq:appen1_lower} and \eqref{eq:appen1_upper}.
Since \eqref{eq:appen1_final} holds for any  small $\epsilon>0$, the claim follows.
$\hfill \square$
\end{proof}

{\em Proof of Theorem \ref{thm3:ratio_of_Rto_PerfCSI}:} ~~~ By
\eqref{eq:appen1_bounds}, \eqref{eq:appen1_lower}, and
\eqref{eq:appen1_upper}, we have
\[
(2q-1-\epsilon)\log M \lesssim_M \Rc_1 \lesssim_M
(2q-1+\epsilon)\log M + \log\log M.
\]
Dividing the above equation by
$\mathbb{E}[\log(1+M\max_k|\alpha_k|^2)]$ and  using the fact that
$\mathbb{E}[\log(1+M\max_k|\alpha_k|^2)] \sim_M \log M + \log\log
M$ \cite{Sharif&Hassibi:05IT}, we have
\[
\frac{(2q-1-\epsilon)\log M}{\log M + \log\log M}\lesssim_M
\frac{\Rc_1}{\mathbb{E}[\log(1+M\max_k|\alpha_k|^2)]} \lesssim_M
\frac{(2q-1+\epsilon)\log M + \log\log M}{\log M + \log\log M}
\]
for arbitrarily and sufficiently small $\epsilon>0$. Hence, we
have
\begin{equation}
2q-1-\epsilon \lesssim_M
\frac{\Rc_1}{\mathbb{E}[\log(1+M\max_k|\alpha_k|^2)]} \lesssim_M
2q-1+\epsilon.
\end{equation}

Now consider the case of $q\in(0,\frac{1}{2})$. In this case,
Eq. \eqref{eq:thm1_asymp_bounds} in Theorem \ref{thm:thm1_asymp_bounds} can
be modified to
\begin{equation}\label{eq:appendB_83}
\beta M^{2q-1-\epsilon} \lesssim_M \mathbb{E}[\log(1+\beta Z)] \lesssim_M \beta M^{2q-1+\epsilon}
\end{equation}
for $\beta>0$ and sufficiently small $\epsilon >0$ such that $2q-1+\epsilon < 0$, since $\log(1+\beta M^{2q-1\pm \epsilon}) \sim_M \beta M^{2q-1\pm \epsilon} $ from $\log(1+x)\to x$ as $x \to 0$. Again applying the law of iterated expectations and
the upper bound in \eqref{eq:appendB_83}, $\Rc_1$ is upper bounded as
\begin{align}
\Rc_1 &\le \mathbb{E}\left[\log\left(1+ \left(\max_k |\alpha_k|^2 \right) Z_{k'} \right) \right] \\
&\lesssim_M {\mathbb{E}}\left(\max_k |\alpha_k^2 |\right) M^{2q-1+\epsilon} \\
&\sim_M (q\log  M) M^{2q-1+\epsilon} \to 0
\end{align}
as $M \to \infty$.
This concludes the proof. $\hfill \square$

\section{}
\label{appen:cone_probability}

A double cone (or cone)
$\Cc_j$ around each axis $j$ in the $M$-dimensional space
is defined as
\begin{equation}
\Cc_j(\eta) = \left\{\hbf_k :
\frac{|\hbf_k^H\ebf_j|}{\|\hbf_k\|}\ge \eta \right\},
\end{equation}
where $\hbf_k \sim \Cc\Nc({\bf 0}, \Ibf_M)$, $\ebf_j$ is the
$j$-th column of the $M \times M$ identity matrix, and $\eta \in
(0,1)$. The probability that the channel vector $\hbf_k$ is
contained in the cone $\Cc_j$ is  given by
\begin{align}
\mathrm{Pr}\{\hbf_k \in \Cc_j(\eta)\}
&= \mathrm{Pr}\{|h_{k,j}| \ge \eta \|\hbf_k\|\} \nonumber\\
&\overset{(a)}{\approx} \mathrm{Pr}\{|h_{k,j}|^2 \ge \eta^2 M \} \nonumber\\
&\overset{(b)}{=}e^{-\eta^2 M}
\end{align}
where $(a)$ becomes tight for large $M$ due to $\|\hbf_k\|^2/M \to
1$, and $(b)$ holds by $|h_{k,j}|^2 \sim \chi^2(2)$. Therefore,
the probability that the cone $\Cc_j$ contains at least one out of
the $K$ channel vectors is given by
\begin{align}
\mathrm{Pr}\{\Cc_j \neq \emptyset\}
&= 1 - \mathrm{Pr}\{\Cc_j = \emptyset\}
= 1 - \mathrm{Pr}\{\hbf_k \notin \Cc_j\}^K  \nonumber \\
&\approx 1 - \left(1-\frac{1}{e^{\eta^2 M}} \right)^K \to \left\{
\begin{array}{ll}
1, & ~~ \text{for} ~~ \lim_{M,K\to\infty}\frac{\log K}{M} = \infty \\
c_1, & ~~ \text{for} ~~ K=\Theta(\exp(\eta^2M)), ~\text{or}~ M = \Theta(\log K) \\
0, & ~~ \text{for} ~~ \lim_{M,K\to\infty}\frac{\log K}{M} = 0
\end{array}
\right.
\end{align}
as $M,K \to \infty$, where $c_1 \in (0,1)$ is a constant. This is
the physical intuition behind the results in
\cite{Sharif&Hassibi:05IT}.




\end{document}

\begin{figure}[ht]
    \centerline{
\scalefig{0.5}
\epsfbox{figures/multibeam_sus_M100to4000_q3_l1to5_rate.eps} }
    \caption{The data rate of $\Rc_S$ with $\alpha_k \sim \Cc\Nc(0,1)$ w.r.t. $M$  when $q=0.3$ (log scale on x-axis)}
    \label{fig:multibeam_sus_rate}
\end{figure}

Finally, we verified the multiple beam RDB with multiple user
selection with the proposed scheduling  method. We set to $q=0.7$
and used $\alpha_k \sim \Cc\Nc(0,1)$, $\forall~k$. Fig.
\ref{fig:multibeam_peruser_RMrate}  shows the per-user rate
$\Rc_{\kappa_m}$ in Theorem \ref{thm:thm_multi_user_selection}
versus $M$ for different $\ell$. It is seen that the per-user rate
$\Rc_{\kappa_m}$ increases when $\ell < 0.4$, whereas it decreases
when $\ell > 0.4$, as $M$ increases, as predicted by Theorem
\ref{thm:thm_multi_user_selection} (i.e., $2q-1-\ell > 0$ or
$2q-1-\ell < 0$).

\begin{figure}[ht]
    \centerline{
\scalefig{0.5}
\epsfbox{figures/multibeam_mus_M100to1000_q7_l1to6_eachrate.eps} }
    \caption{The per-user rate of the multiple beam RDB with the proposed multi-user scheduling
versus $M$: $q=0.7$}
    \label{fig:multibeam_peruser_RMrate}
\end{figure}